\documentclass[10pt, doublecolumn]{IEEEtran}
\usepackage{epsfig,latexsym}
\usepackage{float}
\usepackage{indentfirst}
\usepackage{amsmath}
\usepackage{bm}
\usepackage{amssymb}
\usepackage{times}

\usepackage{algorithm}
\usepackage[noend]{algpseudocode}
\usepackage{subfigure}
\usepackage{psfrag}
\usepackage{hyperref}
\usepackage{cite}
\usepackage{lastpage}
\usepackage{fancyhdr}
\usepackage{color}
 \usepackage{amsthm}
\usepackage{bigints}
\newtheorem{Lemma}{Lemma}
\newtheorem{Corollary}{Corollary}
\newtheorem{lemma}[Lemma]{$\mathbf{Lemma}$}

\newtheorem{corollary}[Corollary]{$\mathbf{Corollary}$}

\newcounter{problem}
\newcounter{save@equation}
\newcounter{save@problem}
\makeatletter
\newenvironment{problem}
{\setcounter{problem}{\value{save@problem}}%
  \setcounter{save@equation}{\value{equation}}%
  \let\c@equation\c@problem
  \subequations
}
{\endsubequations
  \setcounter{save@problem}{\value{equation}}%
  \setcounter{equation}{\value{save@equation}}%
}

\begin{document}
\title{  {    Hybrid NOMA Assisted OFDMA Uplink Transmission    }}

\author{ Zhiguo Ding, \IEEEmembership{Fellow, IEEE}  and H. Vincent Poor, \IEEEmembership{Life Fellow, IEEE}    \thanks{ 
  
\vspace{-2em}

Z. Ding is with Khalifa University, Abu Dhabi, UAE, and University
of Manchester, Manchester, M1 9BB, UK.  
H. V. Poor is with Princeton University, Princeton, NJ 08544, USA.
 

  }\vspace{-3em}}
 \maketitle

\begin{abstract}
Hybrid non-orthogonal multiple access (NOMA) has recently received  significant  research interest due to its ability  to efficiently use resources from different domains and also its compatibility with various orthogonal multiple access (OMA) based legacy networks. Unlike   existing studies on hybrid NOMA  that focus on combining NOMA with time-division multiple access (TDMA), this work considers hybrid NOMA assisted orthogonal frequency-division multiple access (OFDMA). In particular, the impact of a unique feature of hybrid NOMA assisted OFDMA, i.e., the availability of users' dynamic channel state information, on the system performance is analyzed from the following two perspectives. From the optimization perspective,   analytical results are developed which show that with hybrid NOMA assisted OFDMA,  the pure OMA mode is rarely  adopted by the users, and the pure NOMA mode could be optimal for minimizing the users' energy consumption, which    differs from  the hybrid TDMA  case. From the statistical  perspective, two new performance    metrics, namely the power outage probability and the power diversity gain, are developed to quantitatively  measure the performance gain of hybrid NOMA over OMA. The developed analytical results also demonstrate the ability  of hybrid NOMA to meet the users' diverse energy profiles. 
\end{abstract}\vspace{-0.5em}

\vspace{-0.5em} 
 
\section{Introduction}
With the successful rollout of the fifth generation (5G) mobile system, the focus of the   research community has   shifted towards the design of the sixth generation (6G) system \cite{you6g,imt2030}.  In particular, the 6G system is expected to support  important but challenging  services, including ultra-massive machine type communications (umMTC) and enhanced ultra-reliable and low-latency communications (euRLLC), which require improved  use of bandwidth resources and energy \cite{8712527,9369424}.   One of the key enabling 6G techniques for enhancing  both the      spectrum and energy efficiency   is   next generation multiple access (NGMA), which is envisioned to exhibit the following three features:    multi-mode compatibility, multi-domain utilization, and multi-dimensional optimization \cite{9693417, Scizigu,9681910,9681865}.  

Hybrid non-orthogonal multiple access (NOMA) is     a potential NGMA candidate  that can realize the three aforementioned  features simultaneously. Examples of hybrid NOMA include    the   schemes    developed  in \cite{9679390, 10258351, hnomadown}, where  hybrid NOMA was  implemented as an add-on feature of a   time-division multiple access (TDMA)   legacy  network, and    hence its compatibility with the  orthogonal multiple access (OMA) legacy network is guaranteed. The multi-domain utilization feature of hybrid NOMA is due to the fact that a user can  utilize the resources from two domains, namely the time and power domains. In particular, in the time domain,  a user can have access to not only  its own TDMA time slot, but also those time slots which would belong to the other users in the legacy OMA network. Furthermore, in the power domain, the users using the same time slot are assigned to different power levels to avoid multi-access interference. It is worthy pointing out that in hybrid NOMA, a user might choose the NOMA mode in one time slot and the OMA mode in another time slot, and hence hybrid NOMA is  fundamentally different from   schemes such as those in  \cite{ 8823023,9352956,8641304,9556147,9964376}  which force  a user to choose either   OMA or NOMA.   The multi-dimensional optimization feature of hybrid NOMA means that resource allocation over multiple  domains is  carried out jointly.  Compared to conventional  resource allocation,  which is carried out in a single snapshot (e.g., a single time slot in the TDMA case) \cite{8895763,8618435,9525063,10225347}, resource allocation for hybrid NOMA is more challenging, due to its  multi-snap-shot (e.g., multiple time slots in the TDMA case)  nature.

Unlike existing studies of hybrid NOMA which consider TDMA as the OMA protocol, in this paper we consider hybrid    NOMA assisted orthogonal frequency-division multiple access (H-NOMA-OFDMA), which is motivated by the following two considerations. First,   OFDMA is still expected to play an important role in the 6G system, because of its high spectral efficiency and its capability to combat frequency selective fading \cite{docom6g}. And the second  is that in hybrid NOMA assisted TDMA (H-NOMA-TDMA), time diversity cannot be fully exploited due to the following dilemma. For   users whose  channel gains are  constant over multiple time slots, there is no time diversity gain, while  for   users whose  channel gains are   time varying, the time diversity gain indeed becomes available, but a non-causal channel state information (CSI) assumption is required to use such diversity, i.e.,   the user's CSI needs to be available in advance for resource allocation \cite{Scizigu}.  In H-NOMA-OFDMA, a user's channels over different subcarriers are naturally different and can be straightforwardly acquired by the base station, which facilitates the use of     frequency diversity. The main contribution of this paper is to analyze  the impact of this unique feature of H-NOMA-OFDMA, i.e., the availability of users' dynamic CSI, on the system performance by answering the following two questions: 
\begin{itemize}
\item {\it Question 1: Can hybrid NOMA outperform pure OMA and pure NOMA?} This question is answered from an optimization perspective, by first formulating resource allocation for H-NOMA-OFDMA as an optimization problem, and then analyzing the properties of  its optimal solution. In particular, our developed analytical results reveal that the optimal resource allocation in H-NOMA-OFDMA is   fundamentally different from that  of TDMA. For example, the pure OMA solution, i.e., each user adopts the OMA mode, has been shown to be optimal in the H-NOMA-TDMA system, if the time slot durations are not optimized \cite{9679390}. However, in H-NOMA-OFDMA, the pure OMA solution is shown  to be suboptimal, even if the subcarrier bandwidth is not optimized. In addition, in H-NOMA-TDMA, the pure NOMA mode is always suboptimal. However, in the OFDMA case, it is shown  that for some users, the pure NOMA mode is optimal, and the condition for the optimality of pure NOMA is established in this paper. 

\item {\it Question 2: How large is the performance gain of hybrid NOMA over pure OMA?} This question is answered from a statistical  perspective, by first introducing two novel performance    metrics, namely the power outage probability and the power diversity gain. The connections between the new metrics and   conventional ones, such as the outage probability and the diversity gain, are first illustrated \cite{Zhengl03,Laneman04}. Then,  the two metrics are   used to    measure the performance gain of hybrid NOMA over OMA, and  illustrate how the frequency diversity can be utilized by hybrid NOMA to reduce   the users' energy consumption. For example, analytical results are developed in the paper to show that  the power diversity gain of H-NOMA-OFDMA  is proportional to the number of users participating in the NOMA cooperation, whereas the power diversity gains of OMA and H-NOMA-TDMA are always one.  Furthermore, the presented analytical and simulation  results also show that the energy consumption of the users in H-NOMA-OFDMA can be adjusted according to the users' different energy profiles. 

\end{itemize}


\section{System Model}\label{section 2} 
 Consider an OFDMA based legacy network with $M$ uplink users that communicate with the same base station, and are denoted by ${\rm U}_m$, $1\leq m \leq M$, respectively.  Without loss of generality,  assume that each user is assigned to a single subcarrier in the considered legacy network, and denote the   subcarrier  assigned to ${\rm U}_m$ by ${\rm F}_m$, $1\leq m \leq M$.  Each node is assumed to be equipped with a single antenna, and we assume that all users have the same target data rate for their uplink transmission, denoted by $R$.   
 
In this paper, it is assumed that the users have different energy profiles, i.e., some users are more energy constrained than the others.  
 Without loss of generality, we further assume that the users are ordered according to their energy profiles.  In particular, it is assumed that ${\rm U}_M$ is the most energy constrained, e.g., an Internet of Things (IoT) sensor, whereas ${\rm U}_1$ is the least energy constrained, e.g., a device with a reliable energy supply. 
 
   By relying on OFDMA only, the power consumption of ${\rm U}_m$'s uplink transmission is given by $P^{\rm OMA}_m = \frac{e^R-1}{h_{m,m}}$, where  $h_{m,m}=|\tilde{h}_{m,m}|^2$ and  $\tilde{h}_{m,i}$ denotes 
 ${\rm U}_m$'s channel gain on ${\rm F}_i$. In this paper, the users' channel gains on different subcarriers are assumed to be independent and identically distributed (i.i.d.) complex Gaussian random variables with zero means and unit variances. 
 
     \begin{figure*}[t]\centering \vspace{-0em}
    \epsfig{file=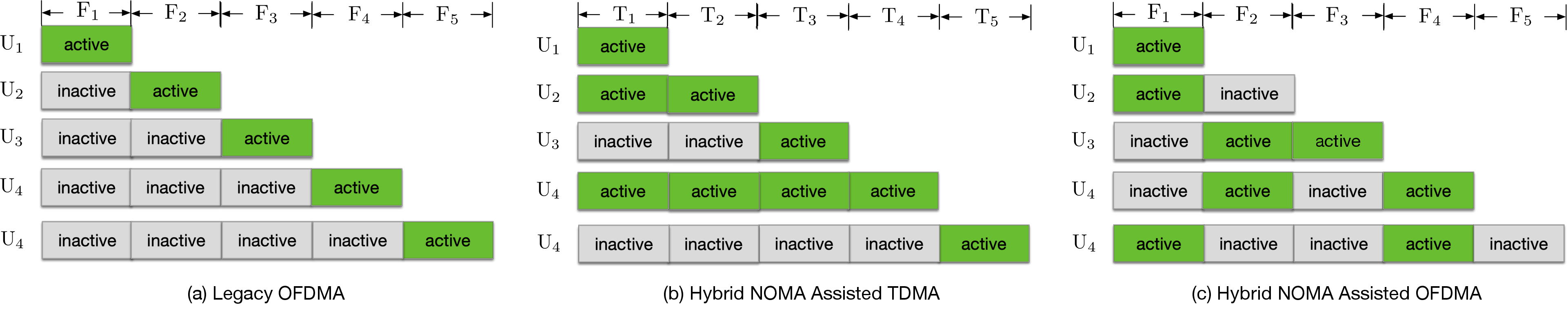, width=0.95\textwidth, clip=}\vspace{-0.5em}
\caption{An illustration of  the hybrid NOMA assisted    OFDMA  network. Compared to the TDMA case studied in  \cite{9679390, hnomadown}, the patten of resource   allocation in the OFDMA case is more complex. For example, in TDMA, a user chooses to use either all the time slots (i.e., the hybrid NOMA mode) or just its own TDMA time slot (i.e., pure OMA), where ${\rm T}_m$ denotes the $m$-th time slot. But in OFDMA, a user, e.g., ${\rm U}_2$, may    adopt the pure NOMA mode, and the time slots a user chooses to use are necessarily continuous.      \vspace{-1em}    }\label{fig0}   \vspace{-0.2em} 
\end{figure*}

The aim of this paper is to reduce the OFDMA users' uplink power consumption by applying hybrid NOMA,   as shown in Fig. \ref{fig0}.  
 In particular, each user can still use the subcarrier that is allocated to it in the  OFDMA legacy system, but a more energy constrained user is provided the access to more subcarriers. For example,    ${\rm U}_m$ can still use its own OFDMA subcarrier,  ${\rm F}_m$, but also has access to    ${\rm F}_i$, $1\leq i \leq m$.  As such, the most energy constrained user, ${\rm U}_M$, can have access to all subcarriers, whereas the least energy constrained user, ${\rm U}_1$, can use a single subcarrier only. 

Therefore, at   subcarrier ${\rm F}_n$, the base station receives the signals simultaneously transmitted by ${\rm U}_i$, $n\leq i \leq M$. The base station  carries out successive interference cancellation (SIC)  by employing  the users' diverse energy profiles, e.g.,  at   subcarrier ${\rm F}_n$, ${\rm U}_M$'s signal is decoded first and ${\rm U}_1$'s signal is decoded last. With the used SIC decoding strategy, and assuming additive white Gaussian noise (AWGN),   ${\rm U}_m$'s signal on  ${\rm F}_n$ can be decoded with the following achievable data rate:
\begin{align}
R_{m,n} =\log\left(
1+
\frac{h_{m,n}P_{m,n}}{\sum^{m-1}_{j=n}h_{j,n}P_{j,n}+1}
\right),
\end{align} 
where $P_{m,n}$ denotes ${\rm U}_m$'s transmit power on  ${\rm F}_n$, and the noise power is assumed to be normalized. 
It is worthy  pointing  out that with the used SIC strategy, ${\rm U}_m$'s achievable data rate on ${\rm F}_m$ is simply $\log\left(
1+
 {h_{m,m}P_{m,m}} 
\right)$, which is identical to the user's data rate in OFDMA. 

Because of the energy constrained nature of the considered uplink scenario, the following power  minimization problem is considered in this paper:
  \begin{problem}\label{pb:1} 
  \begin{alignat}{2}
\underset{P_{m,n}\geq0   }{\rm{min}} &\quad   \sum^{M}_{m=1}  \sum^{m}_{n=1}P_{m,n}\label{1tst:1}
\\ s.t. &\quad  \sum^{m}_{n=1} R_{m,n}\geq R, \quad 1\leq m \leq M.  \label{1tst:2}  
  \end{alignat}
\end{problem}  

{\it Remark 1:} Similar to the hybrid NOMA schemes considered in  \cite{9679390, hnomadown}, problem \eqref{pb:1} leads to a general transmission framework, where the modes of pure OMA and pure NOMA can be obtained as   special cases by adopting different choices of the power allocation coefficients, $P_{m,n}$. 

{\it Remark 2:} The optimal solution of problem \eqref{pb:1} leads to a more complex resource allocation patten than its counterpart in the TDMA case, as illustrated in Fig. \ref{fig0}. For example, in H-NOMA-TDMA, the pure NOMA mode is always suboptimal, and  a user always uses the bandwidth resource block (i.e, the time slot) that is assigned to it  in OMA. However, in H-NOMA-OFDMA, the use of pure NOMA may outperform both hybrid NOMA and pure OMA. Another example is that for the uplink TDMA case with equal-duration time slots, i.e., the  time slot durations are not optimized, the use of pure OMA is always optimal.  In the OFDMA case, the pure OMA solution is rarely optimal,  particularly if there is a large number of users, as will  be shown in the next section. 

{\it Remark 3:} The fact that  a user's channel gains on different subcarriers are different, i.e., for ${\rm U}_m$, $h_{m,1}\neq \cdots \neq h_{m,m}$, leads to the   availability of the frequency diversity \cite{Rappaport}. The proposed hybrid NOMA scheme can effectively utilize the frequency diversity inherent in the OFDMA legacy network, and ensures that a user that is  more energy constrained  consumes less energy   for its uplink transmission, as shown in the next section. It is important to note that in TDMA, there is a dilemma   associated with using such diversity as noted in the introduction.  In particular, it is preferable to exploit such diversity as it can significantly reduce energy consumption. However, the use of such diversity requires the base station to know the users' future CSI, which is challenging in practice.  



\section{Performance Analysis: An Optimization Perspective }
In this section, the properties   of H-NOMA-OFDMA transmission are analyzed from an optimization perspective. 
In particular,   tools from   convex optimization   are used to characterize  the important features of  the optimal solution of problem \eqref{pb:1}, and address the issues of whether  hybrid NOMA can outperform pure OMA and pure NOMA.


\subsection{Optimality of Pure OMA}\label{subIIB}
It is straightforward to show that problem \eqref{pb:1} is not a convex optimization problem since the constraints in \eqref{1tst:1} are not convex.  Therefore, it is challenging to obtain a closed-form expression for the optimal solution of problem \eqref{pb:1}. However, it is important to point out that the   Karush–Kuhn–Tucker (KKT) conditions can still  be used as an optimality certificate for the   solutions of this non-convex optimization  problem \cite{Boyd}. By using this fact,   insights   into H-NOMA-OFDMA transmission can be obtained, as shown in the following lemma. 
\begin{lemma}\label{lemma1}
A necessary condition for the pure OMA solution, i.e., $P_{m,m} = \frac{e^R-1}{h_{m,m}}$ and $P_{m,n}=0$ for $m\neq n$,  to be the optimal solution of problem \eqref{pb:1} is given by
\begin{align}
h_{m,n}\leq h_{m,m}, \quad 1\leq n \leq m-1, 1\leq m \leq M. 
\end{align}
\end{lemma}
\begin{proof}
See Appendix \ref{proof1}.
 \end{proof}
 
The use of Lemma \ref{lemma1} immediately leads to the following corollary regarding the optimality of    the pure OMA solution   in practice. 
\begin{corollary}\label{corollary1}
Assuming that the users' channel gains at different subcarriers are i.i.d., the probability that the pure OMA mode   is optimal approaches  zero if $M\rightarrow \infty$.   
\end{corollary}  
\begin{proof}
Lemma \ref{lemma1} shows that a necessary condition for the pure OMA solution to be optimal is given by $h_{m,n}\leq h_{m,m}$, for $1\leq n \leq m-1$ and $ 1\leq m \leq M$. Because the users' channel gains are assumed to be i.i.d., the probability of the event  that this condition holds is given by
\begin{align}\nonumber 
\mathbb{P}^{\rm OMA} = &\mathbb{P}\left(
h_{m,n}\leq h_{m,m}, \quad 1\leq n \leq m-1, 1\leq m \leq M
\right)\\\nonumber 
= &\prod^{M}_{m=2}\prod^{m-1}_{n=1} \mathbb{P}\left(
h_{m,n}\leq h_{m,m}  
\right)  \\  
= &\prod^{M}_{m=2}\frac{1}{2^{m-1}} = \frac{1}{2^{\frac{M(M-1)}{2}}}\rightarrow 0,
\end{align}
if $M\rightarrow \infty$. The proof of the corollary is complete. 
\end{proof}
{\it Remark 4:} Corollary \ref{corollary1} illustrates an important difference between the applications of hybrid NOMA to   TDMA and OFDMA networks. Recall that   for the uplink TDMA case, if the time slot durations are not optimized, pure OMA outperforms both hybrid NOMA and pure NOMA \cite{9679390, hnomadown}. However, for the uplink OFDMA case, Corollary \ref{corollary1} demonstrates that the NOMA schemes outperform pure OMA, particularly for   situations in which the number of users is very large.  

 \subsection{Optimality of Pure NOMA} 
As noted in the introduction, in studies of hybrid NOMA assisted  TDMA, pure NOMA has been shown to be suboptimal  \cite{9679390, hnomadown}. However, in the context of OFDMA, the use of the pure NOMA mode can be optimal, as shown in the following. Due to the complex and random patterns of the hybrid NOMA subcarrier allocation, a closed-form expression for the pure NOMA solution cannot be obtained. For example, consider a particular pure NOMA solution,  where all the users use ${\rm F}_1$ only, i.e.,  $P_{m,1}\neq 0$, for $1\leq m \leq M$, and the users' remaining power allocation coefficients are zero. For this particular pure NOMA case, problem \ref{pb:1} can be recast as follows:
  \begin{problem}\label{pb:2} 
  \begin{alignat}{2}
\underset{P_{m,n}\geq0   }{\rm{min}} &\quad   \sum^{M}_{m=1} P_{m,1} 
\\ s.t. &\quad \log\left(
1+
\frac{h_{m,1}P_{m,1}}{\sum^{m-1}_{j=1}h_{j,1}P_{j,1}+1}
\right) \geq R,  1\leq m \leq M\nonumber ,
  \end{alignat}
\end{problem}  
which is equivalent to the following form: 
  \begin{problem}\label{pb:3} 
  \begin{alignat}{2}  
\underset{P_{m,n}\geq0   }{\rm{min}} &\quad   \sum^{M}_{m=1} P_{m,1} 
\\ s.t. &\quad \nonumber
 h_{m,1}P_{m,1} 
  \geq \phi \sum^{m-1}_{j=1}h_{j,1}P_{j,1}+\phi ,  1\leq m \leq M   ,
  \end{alignat}
\end{problem}  
where $\phi=e^R-1$.  Problem \eqref{pb:3} is a linear program that  can be solved numerically, but it is challenging to obtain a closed form expression for its optimal solution.   Therefore, in the following, we consider the case  in which   the first $M-1$ users choose the OMA mode, and the last user ${\rm U}_M$ chooses the pure NOMA mode, i.e., its dedicated OMA subcarrier is not used. Because  the first $M-1$ users choose the OMA mode, problem \eqref{pb:1} can be simplified as follows:
  \begin{problem}\label{pb:4} 
  \begin{alignat}{2}
\underset{1\leq n \leq M-1 }{\text{arg } }\text{min} &\quad    P_{M,n}\label{4tst:1}
\\ s.t. &\quad   R_{M,n}\geq R.  \label{4tst:2}  
  \end{alignat}
\end{problem}  
Recall that $R_{M,n} =\log\left(
1+
\frac{h_{M,n}P_{M,n}}{\sum^{M-1}_{j=n}h_{j,n}P_{j,n}+1}
\right)$. If the first $M-1$ users choose the OMA mode, $R_{M,n}$ can be simplified as follows:
\begin{align}
R_{M,n} =\log\left(
1+
\frac{h_{M,n}P_{M,n}}{ h_{n,n}P_{n,n}+1}
\right)=\log\left(
1+
  \frac{h_{M,n}P_{M,n} }{e^R}
\right).
\end{align}

Because ${\rm U}_M$ suffers the same amount of interference during the first $M-1$ subcarriers, i.e., $h_{n,n}P_{n,n}+1=e^{R}$, for $1\leq n\leq M-1$, the optimal solution of problem \eqref{pb:4} is simply $n^*=\underset{1\leq n \leq M-1 }{\text{arg max} h_{M,n} }$, i.e., the user chooses the best subcarrier.  Without loss of generality, assume that $h_{M,1}\geq h_{M,i}$, $2\leq i\leq M$, which means  $n^*=1$ and the corresponding pure NOMA solution is $P_{m,m}=\frac{e^R-1}{h_{m,m}}$, $1\leq m \leq M-1$, $P_{M,1}=\frac{e^R(e^R-1)}{h_{M,1}}$, and the other power coefficients are zero. 

 \begin{lemma}\label{lemma2}
 A necessary conditions for the considered pure NOMA solution to be optimal to problem \eqref{pb:1} are give by 
  \begin{align}\label{lemma2x}
\left\{\begin{array}{l}  
 \frac{  h_{m,n }  }{ h_{m,m } }\leq 1
 , \quad 1\leq n\leq m \leq M-1 \\
   \frac{ h_{M,n } }{ h_{M,1 } }\leq 1
 ,  \quad 2\leq n\leq M-1\\
 e^R  h_{M,M}\leq h_{M,1 }
 \end{array}\right.. 
 \end{align}
 \end{lemma}
 \begin{proof}
 See Appendix \ref{proof2}.
\end{proof}

 We note that the considered pure NOMA solution is obtained with the condition that 
 $h_{M,1}\geq h_{M,i}$, $2\leq i\leq M$, which is weaker than the one in  Lemma \ref{lemma2}, since $ e^R  h_{M,M}\leq h_{M,1 }$ implies $   h_{M,M}< e^R  h_{M,M}\leq h_{M,1 }$.
 
{\it Remark 5:} Lemma \ref{lemma2} shows that it is possible for a user to adopt the pure NOMA model in H-NOMA-OFDMA. Lemma \ref{lemma2} is also useful to illustrate another important difference between  hybrid NOMA assisted OFDMA and TDMA. For the TDMA case, a user's channel does not change over the $M$ time slots, and therefore,  Lemma \ref{lemma2} is still applicable, but with the changes that $h_{m,1}=\cdots =h_{m,m}$. Therefore, the   conditions in \eqref{lemma2x} can be simplified as follows:
 \begin{align}\label{lemma2x}
\left\{\begin{array}{l}  
1\leq 1
 , \quad 1\leq n\leq m \leq M-1 \\
 1\leq 1
 ,  \quad 2\leq n\leq M-1\\
 e^R  \leq 1
 \end{array}\right..
 \end{align}
Because the last condition can never be met,   a user will never choose the pure NOMA mode in the TDMA case, a conclusion   consistent to those made in \cite{9679390, hnomadown}. 

\subsection{Successive Resource Allocation}
Similar to \cite{9679390, hnomadown}, a low-complexity   algorithm can be developed to find  a sup-optimal solution of problem \eqref{pb:1}, by utilizing the successive nature of the SIC decoding strategy at   the base station. In particular, the proposed successive resource allocation algorithm consists of $M$ stages, where during the $m$-th stage, the following optimization problem is solved:
  \begin{problem}\label{pb:6} 
  \begin{alignat}{2}
\underset{P_{m,n}\geq0   }{\rm{min}} &\quad    \sum^{m}_{n=1}P_{m,n}\label{6tst:1}
\\ s.t. &\quad  \sum^{m}_{n=1} \log\left(
1+
\frac{h_{m,n}P_{m,n}}{\sum^{m-1}_{j=n}h_{j,n}P_{j,n}+1}
\right)\geq R,   \label{6tst:2}  
  \end{alignat}
\end{problem}  
where the parameters, $P_{i,n}$, $i<m$, are obtained from the previous stages. Because $P_{i,n}$, $i<m$, are fixed, it is straightforward to show that problem \eqref{pb:6} is a convex optimization problem, and hence can be efficiently solved by those off-shelf optimization solvers. 

Unlike the hybrid NOMA cases in \cite{9679390, hnomadown}, a closed form expression for  the successive resource allocation solution is challenging to obtain in the general case. However, insightful understandings can be obtained for the two-user special case. In particular, 
the optimality of the proposed successive resource allocation algorithm in the two-user special case is shown   in the following lemma.
\begin{lemma}\label{lemma3}
For the two-user special case, the solution obtained by the proposed successive resource allocation algorithm is optimal for problem \eqref{pb:1}.
\end{lemma}
\begin{proof}
  This lemma can be proved  by the method of contradiction. Denote the optimal solutions of problem \eqref{pb:1} by $P_{1,1}^*$, $P_{2,1}^*$ and $P_{2,2}^*$, and the solutions of successive resource allocation by $P_{1,1}^{\rm SRA}$, $P_{2,1}^{\rm SRA}$ and $P_{2,2}^{\rm SRA}$, respectively. Assume that $P_{1,1}^*+P_{2,1}^*+P_{2,2}^*<P_{1,1}^{\rm SRA}+P_{2,1}^{\rm SRA}+P_{2,2}^{\rm SRA}$. 
  
  First note  that by  using   the proposed resource allocation algorithm, ${\rm U}_1$'s power allocation is simply its OMA choice, i.e., $P_{1,1}^{\rm SRA} = \frac{e^R-1}{h_{1,1}}$, which is the minimal transmit power for ${\rm U}_1$ to meet the target data rate $R$. Therefore, $P_{1,1}^{\rm SRA} \leq P_{1,1}^*$, which means that  $ P_{2,1}^*+P_{2,2}^*< P_{2,1}^{\rm SRA}+P_{2,2}^{\rm SRA}$ in order to ensure $P_{1,1}^*+P_{2,1}^*+P_{2,2}^*<P_{1,1}^{\rm SRA}+P_{2,1}^{\rm SRA}+P_{2,2}^{\rm SRA}$. However,   $P_{2,1}^{\rm SRA}$ and $P_{2,2}^{\rm SRA}$   are   optimal solutions of the following optimization problem:
  \begin{problem}\label{pb:7} 
  \begin{alignat}{2}
\underset{P_{2,n}\geq0   }{\rm{min}} &\quad    P_{2,1}+P_{2,2}\label{7tst:1}
\\ s.t. &\quad    \log\left(
1+
\frac{h_{2,1}P_{2,1}}{ h_{1,1}P_{1,1}^{\rm SRA}+1}
\right)+ \log\left(
1+
 h_{2,2}P_{2,2} 
\right)\geq R.   \nonumber  
  \end{alignat}
\end{problem}
Comparing problem \eqref{pb:7} to the following problem:
  \begin{problem}\label{pb:7dd} 
  \begin{alignat}{2}
\underset{P_{2,n}\geq0   }{\rm{min}} &\quad    P_{2,1}+P_{2,2}\label{7tst:1}
\\ s.t. &\quad    \log\left(
1+
\frac{h_{2,1}P_{2,1}}{ h_{1,1}P_{1,1}^{*}+1}
\right)+ \log\left(
1+
 h_{2,2}P_{2,2} 
\right)\geq R,   \nonumber  
  \end{alignat}
\end{problem}
one can notice that the feasibility region of problem \eqref{pb:7dd} is no larger than that of \eqref{pb:7} since $P_{1,1}^{\rm SRA}\leq P_{1,1}^{*}$. Therefore,   $ P_{2,1}^*+P_{2,2}^*< P_{2,1}^{\rm SRA}+P_{2,2}^{\rm SRA}$ cannot be true. The proof of the lemma is complete. 
\end{proof}

By using Lemma \ref{lemma3}, a closed-form expression for the optimal solution of problem \eqref{pb:1} can be obtained for the two-user special case, as shown in the following lemma. 

\begin{lemma}\label{lemma4}
For the two-user special case, an  optimal solution of problem \eqref{pb:1} can be obtained as follows:
 \begin{align}
\left\{\begin{array}{ll}
\text{Pure OMA mode:} &    \text{if } \frac{ h_{2,2}}{ h_{2,1}} \geq 1\\
\text{Pure NOMA mode:} &     \text{if }   \frac{ h_{2,2}}{ h_{2,1}}  \leq e^{-2R} \\
\text{Hybrid NOMA mode:}&   \text{if }   e^{-2R}\leq  \frac{ h_{2,2}}{ h_{2,1}}  < 1  \\
 \end{array}\right.
 \end{align}
 where $P_{1,1}^{\rm *}=  \frac{e^R-1}{h_{1,1}}$, for the hybrid NOMA   mode, $
P_{2,1}^{\rm *} =\frac{e^{R}}{ \sqrt{ h_{2,1}   h_{2,2}}}    -\frac{e^{R}}{ h_{2,1}}$ and $
P_{2,2}^{\rm *} =\frac{e^{R}}{ \sqrt{ h_{2,1}   h_{2,2}}}    -\frac{1}{ h_{2,2}}$, for the pure NOMA mode, 
$P_{2,1}^{\rm *}=\frac{e^R(e^{R}-1)}{h_{2,1}}$ and $P_{2,2}^{\rm *}=0$, for the pure OMA mode, 
$P_{2,1}^{\rm *}=0$ and $P_{2,2}^{\rm *} = \frac{e^{R}-1}{h_{2,2}}$.  
 
\end{lemma}
 \begin{proof}
 See Appendix \ref{proof4}. 
 \end{proof}
{\it Remark 6:}    Lemma \ref{lemma4}   demonstrates how challenging it can be to obtain closed-form expressions for optimal H-NOMA-OFDMA transmission, since even for the simple two-user special case, there are three possible expressions.



\section{Performance Analysis: A  Statistical Perspective}
 In the previous  section, the features of H-NOMA-OFDMA have been revealed from the optimization perspective, where    optimality conditions  have been established for different transmission modes. These results are useful to understand  under which conditions  hybrid NOMA outperforms the other transmission modes, particularly pure OMA. This section focuses on  the following question - how large can  the performance gain of hybrid NOMA over pure OMA   be?      In particular, two new criteria, namely the power outage probability and the power diversity gain, are introduced to quantitize  the performance gain of hybrid NOMA over OMA. 

Formally, the power outage probability is defined as the probability  that, in order to meet the target data rate, a user needs more transmit power than its power budget, denoted by $\rho$.  Take OMA transmission as an example, where the power outage probability for ${\rm U}_m$ is given by
\begin{align}
\mathbb{P}^{\rm out}_m =& \mathbb{P}\left(P_m^{\rm OMA}\geq \rho\right) = \mathbb{P}\left(\frac{e^R-1}{h_{m,m}}\geq \rho\right) \\\nonumber =& \mathbb{P}\left(\log\left(1+\rho h_{m,m}\right)\leq R\right),
\end{align}
which shows that for OMA, the power outage probability is equivalent to the standard  outage probability \cite{Zhengl03}. Similar to the conventional outage diversity gain    \cite{Laneman04}, ${\rm U}_m$'s  power diversity gain can be defined as follows:
\begin{align}
d = - \underset{\rho\rightarrow \infty}{\lim} \frac{\log \mathbb{P}_m^{\rm out}}{\log \rho}.
\end{align}
These two proposed   metrics are useful   for measuring  how the users' predefined  energy profiles can be met in dynamic  fading scenarios. For example, with a smaller power outage probability and a  larger power diversity gain, a communication  system is more robust to meet the users' predefined  energy profiles. 

Two case studies will be carried out in this section by using these  two performance evaluation metrics  to illustrate  the performance gain of hybrid NOMA over OMA.

\subsection{The Two-User Special Case}
We first consider the special case of two users, because  a closed-form expression for the optimal solution of problem \eqref{pb:1} is available in this case. Note that with H-NOMA-OFDMA, ${\rm U}_1$'s transmit power is the same as its OMA transmit power. Therefore, only ${\rm U}_2$'s   transmit power is considered  in the following, where the corresponding power outage probability is given by 
\begin{align}
\mathbb{P}^{\rm out}_2 = \mathbb{P}\left(P_{2,1}+P_{2,2}\geq\rho \right) .
\end{align}
The following lemma provides a high signal-to-noise ratio  (SNR) approximation of  ${\rm U}_2$'s power outage probability.

\begin{lemma} \label{lemma5}
At high SNR, i.e., $\rho \rightarrow \infty$,  ${\rm U}_2$'s power outage probability can be approximated as follows:
 \begin{align}
\mathbb{P}^{\rm out} \approx&         \frac{(e^{R}-1)^2}{\rho ^2}   +\frac{\tau}{\rho^2}  ,
\end{align}
where $\tau$ is a constant not related to $\rho$, $\tau =\int^{e^R}_{1}\frac{y-1}{\rho} \left(\min\left\{
e^{2R} ,g_2(y) 
\right\} -\max\left\{1 ,g_1(y)\right\}  
\right)dy $, $g_1(y) = \frac{2e^{2R}     -e^{R} y -2e^{R}\sqrt{e^{2R}  -e^{R}  y} }{ y^2}$, and $g_2(y) = \frac{2e^{2R}     -e^{R} y +2e^{R}\sqrt{e^{2R}  -e^{R}  y} }{ y^2}$
\end{lemma}
\begin{proof}
See Appendix \ref{proof5}.
\end{proof}
 By using Lemma \ref{lemma5}, the following corollary for ${\rm U}_2$'s power diversity order can be straightforwardly  obtained.
 \begin{corollary}\label{corollary2}
By using hybrid NOMA assisted OFDMA,  ${\rm U}_2$'s power diversity order is $2$, where the power diversity order of pure OFDMA is $1$. 
 \end{corollary}

 {\it Remark 7:} Lemma \ref{lemma5} and  Corollary \ref{corollary2}  demonstrate that the use of hybrid NOMA can significantly reduce the OFDMA uplink energy consumption, and more robustly meet the users' stringent energy constraints.  The reason for this performance gain can be explained by the following extreme example. Suppose  that ${\rm U}_2$'s channel gain on ${\rm F}_2$ is in deep fading, i.e., $h_{2,2}\rightarrow 0$. The use of OMA leads to a singular  energy situation, i.e., the user needs to use infinite transmit power. However, by using hybrid NOMA, this singularity issue can be avoided, since  ${\rm U}_2$ can simply discard ${\rm F}_2$ and rely on ${\rm F}_1$ only.  
  Following   steps similar to those in the proof of Lemma \ref{lemma5} and using the fact that $h_{2,1}=h_{2,2}$ for the TDMA case, the following corollary can be obtained.  
  \begin{corollary}\label{corollary3}
By using hybrid NOMA assisted TDMA,  ${\rm U}_2$'s power diversity order is   $1$. 
 \end{corollary}
 Corollary  \ref{corollary3} illustrates another difference between the applications of hybrid NOMA to TDMA and OFDMA systems, which is due to the    use of the users' dynamic channel conditions.  

 \subsection{A Multi-User Special Case}
 Because a closed-form expression for the optimal solution of problem \eqref{pb:1} is not available for the general multi-user case, we focus on the   pure NOMA solution considered in  Section \ref{subIIB}, which offers important insights into general hybrid NOMA transmission. In particular,   assume    that the first $M-1$ users adopt the OMA mode, and ${\rm U}_M$ chooses ${\rm F}_1$ only, where it is assumed that $h_{M,1}\geq h_{m,i}$, $2\leq i\leq M$. Therefore, the corresponding power allocation solution is $P_{m,m}=\frac{e^R-1}{h_{m,m}}$, $1\leq m \leq M-1$, $P_{M,1}=\frac{e^R(e^R-1)}{h_{M,1}}$, and the other power coefficients are zero. 
 
 
 With this considered  NOMA solution, ${\rm U}_M$'s power outage probability is given by
 \begin{align}
 \mathbb{P}_M^{\rm out} =&  \mathbb{P}\left(
 \frac{e^R(e^R-1)}{h_{M,1}}\geq \rho, e^R  h_{M,M}\leq h_{M,1 }, \right.\\\nonumber &\left. h_{M,1}\geq h_{m,i},2\leq i\leq M
 \right),
 \end{align}
 where the condition in Lemma \ref{lemma2} is used. 
Again applying the complex Gaussian distribution assumption, the users' unordered  channel gains are i.i.d. exponentially distributed. Because $h_{M,1}\geq h_{m,i}$, $2\leq i\leq M-1$, the probability density function (pdf) of $h_{M,1}$ is given by $f(x) =(M-1)\left(1-e^{-x}\right)^{M-2}e^{-x} $. Therefore, ${\rm U}_M$'s power outage probability can be obtained as follows:  
  \begin{align}
 \mathbb{P}_M^{\rm out} =&  \mathbb{P}\left(
 e^R  h_{M,M}\leq h_{M,1 }\leq  \frac{e^R(e^R-1)}{\rho}, \right.\\\nonumber &\left. h_{M,1}\geq h_{m,i},2\leq i\leq M
 \right)\\\nonumber =& \int^{\frac{  e^R-1}{\rho}}_{0} \hspace{-0.6em}e^{-y}\int^{ \frac{e^R(e^R-1)}{\rho}}_{e^R y}  \hspace{-0.6em}(M-1)\left(1-e^{-x}\right)^{M-2}e^{-x} dxdy.
 \end{align}
 With some straightforward algebraic manipulations, the outage probability can be obtained as follows:
   \begin{align}
 \mathbb{P}_M^{\rm out}   
   =& - \left(1-e^{-\frac{e^R(e^R-1)}{\rho}}\right)^{M-1} \left(1 - e^{-\frac{  e^R-1}{\rho}}\right)  \\\nonumber &+
  \sum^{M-1}_{i=0}{M-1\choose i}(-1)^i    \left(1-  e^{-\left(1+ie^R \right)\frac{  e^R-1}{\rho}}\right) .
 \end{align}
 Following   steps similar to those in the proof for Lemma \ref{lemma5}, ${\rm U}_M$'s power outage probability can be approximated  as follows:
   \begin{align}
 \mathbb{P}_M^{\rm out}   
      \approx &   \frac{e^{(M-1)R}(e^R-1)^{M}}{\rho^{M}}    +
     \frac{e^{(M-1)R}}{M} \frac{  (e^R-1)^M}{\rho^M} ,
     \end{align}
     which means that the power diversity order experienced by ${\rm U}_M$ is $M$. Recall that    the power diversity order of OMA is only $1$, which means that the use of hybrid NOMA can significantly reduce ${\rm U}_M$'s energy consumption.  Our simulation results show that  by using  hybrid NOMA,  ${\rm U}_i$'s power diversity order is larger than ${\rm U}_j$'s power diversity order, for $i>j$, which   not only justifies   the adopted SIC decoding order, but also illustrates the capability of hybrid NOMA to meet the users' different energy profiles.  
 
      \begin{figure}[t]\centering \vspace{-0em}
    \epsfig{file=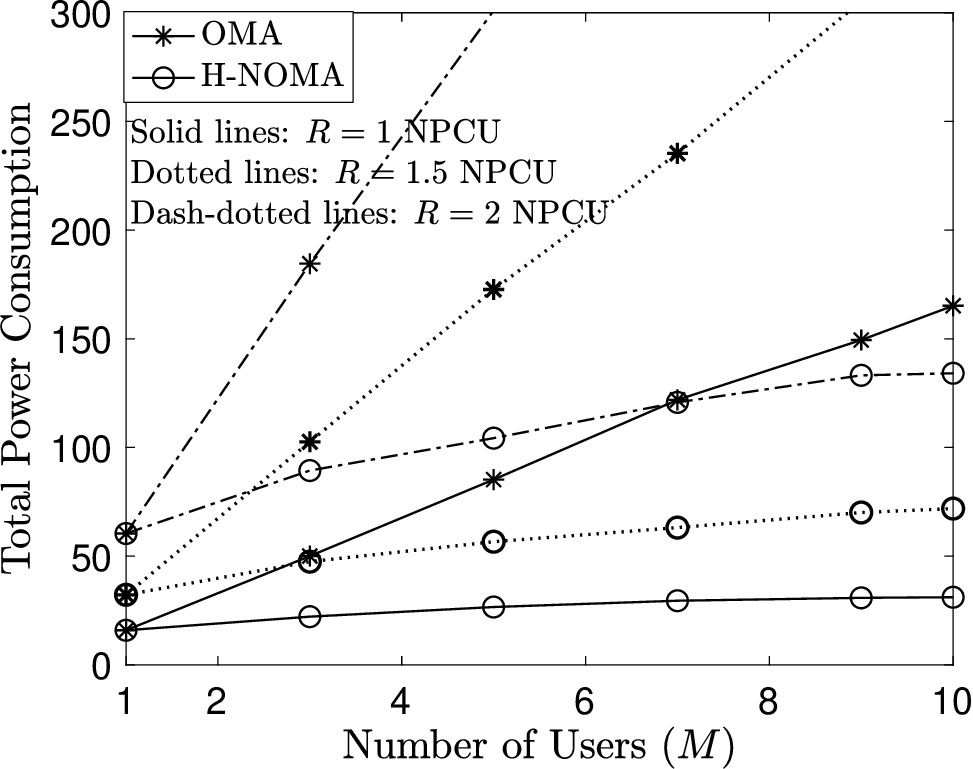, width=0.45\textwidth, clip=}\vspace{-0.5em}
\caption{The total power consumption achieved by   hybrid NOMA assisted OFDMA, where the users' channel gains are assumed to be complex Gaussian distributed with zero mean and unit variance. NPCU denotes   nats per channel use, and $\epsilon=10^{-5}$.    \vspace{-1em}    }\label{fig1}   \vspace{-0.5em} 
\end{figure}

\section{Numerical Results}
In this section, the performance of H-NOMA-OFDMA is demonstrated using computer simulation, where the accuracy of the developed analytical results is also examined. Two metrics, namely the averaged total power consumption and the power outage probability, are used for performance evaluation in the following two subsections, respectively. NPCU denotes nats per channel use. 
\vspace{-1em}
\subsection{Averaged Total Power Consumption}
In Fig. \ref{fig1}, the performance of H-NOMA-OFDMA is shown as a function of the number of the users, $M$. Recall that the users' channels are assumed to be i.i.d. complex Gaussian random variables with zero means and unit variances, which means that   the averaged   power consumption for both OMA and NOMA will be infinite, as illustrated in the following. Take   OMA as an example, where  $\mathcal{E}_{h_{m,m}}\left\{  P_{m}^{\rm OMA}\right\} =\int^{\infty}_{0} \frac{e^R-1}{x} e^{-x}dx   \rightarrow \infty$. In order to avoid this singularity  issue, it is assumed that $h_{m,n}\geq \epsilon$, where $\epsilon=10^{-5}$ is used in Fig. \ref{fig1}. As can be seen from the figure, for the case of $M>1$, the use of hybrid NOMA can significantly reduce the power consumption of OFDMA systems. The performance gain of NOMA over OMA can be further increased by increasing  the target data rate. Furthermore,  Fig. \ref{fig1} shows that   the total power consumption of OMA grows linearly with the number of the users, since $\mathcal{E}_{h_{m,m}}\left\{ \sum^{M}_{m=1} P_{m}^{\rm OMA}\right\} = M\mathcal{E}_{h_{m,m}}\left\{  P_{m}^{\rm OMA}\right\}$. However by using hybrid NOMA,  there is only a slight increase  in power consumption  when the number of the users is increased. Therefore, the performance gain of NOMA over OMA becomes significantly  larger if  there are more users participating in the NOMA cooperation, which makes H-NOMA-OFDMA  particularly  appealing for  umMTC.

      \begin{figure}[t]\centering \vspace{-0em}
    \epsfig{file=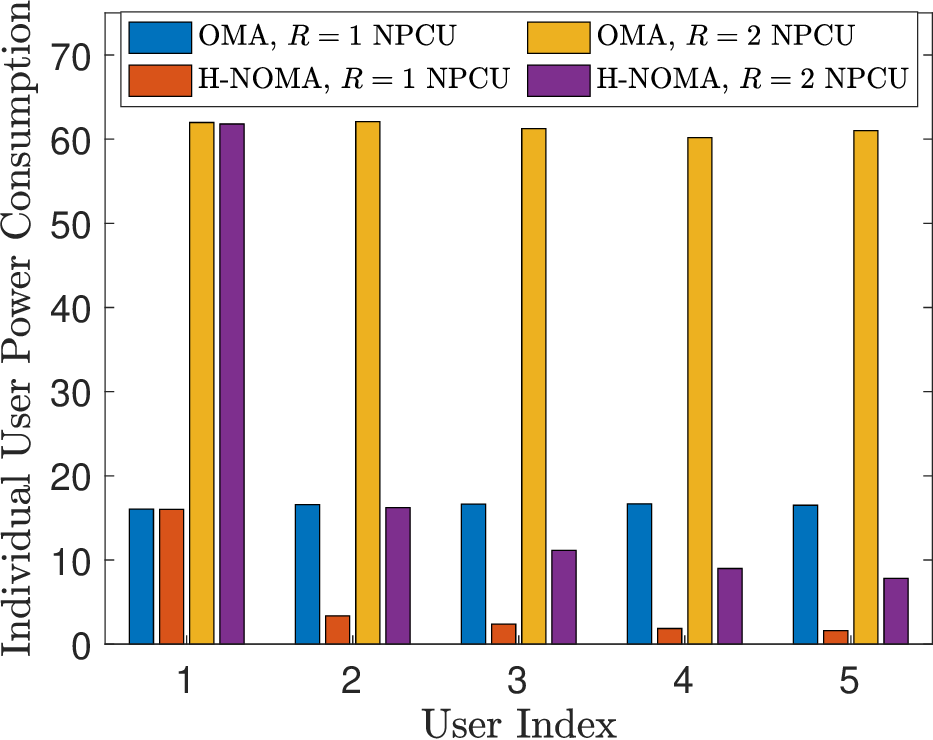, width=0.45\textwidth, clip=}\vspace{-0.5em}
\caption{The users' individual power consumption achieved by hybrid NOMA assisted OFDMA, where   $\epsilon=10^{-4}$ and $M=5$.    \vspace{-1em}    }\label{fig2}   \vspace{-0.1em} 
\end{figure}

In addition to its ability  to reduce the total power consumption, H-NOMA-OFDMA can also meet the users' diverse energy profiles, as shown in Fig. \ref{fig2}. Recall that in this study,  ${\rm U}_m$ is assumed to be more energy constrained than  ${\rm U}_n$, $n<m$. Fig. \ref{fig2} demonstrates that  the use of hybrid NOMA indeed ensures that ${\rm U}_m$ consumes less transmit power than ${\rm U}_n$, $n<m$, whereas with OMA,   all the users use the same transmit power. This benefit is due to the fact that  the use of  hybrid NOMA ensures that  a more energy constrained user has     access to more bandwidth resources than a less energy constrained user.   For example,   ${\rm U}_M$ has access to all   $M$ subcarriers, and hence has more degrees of freedom to reduce its energy consumption, compared to the other users.  Fig. \ref{fig2} also shows that   the use of hybrid NOMA yields less power consumption than OMA, for all   users except ${\rm U}_1$, which is consistent with the observation in Fig. \ref{fig1}.  

      \begin{figure}[t]\centering \vspace{-0em}
    \epsfig{file=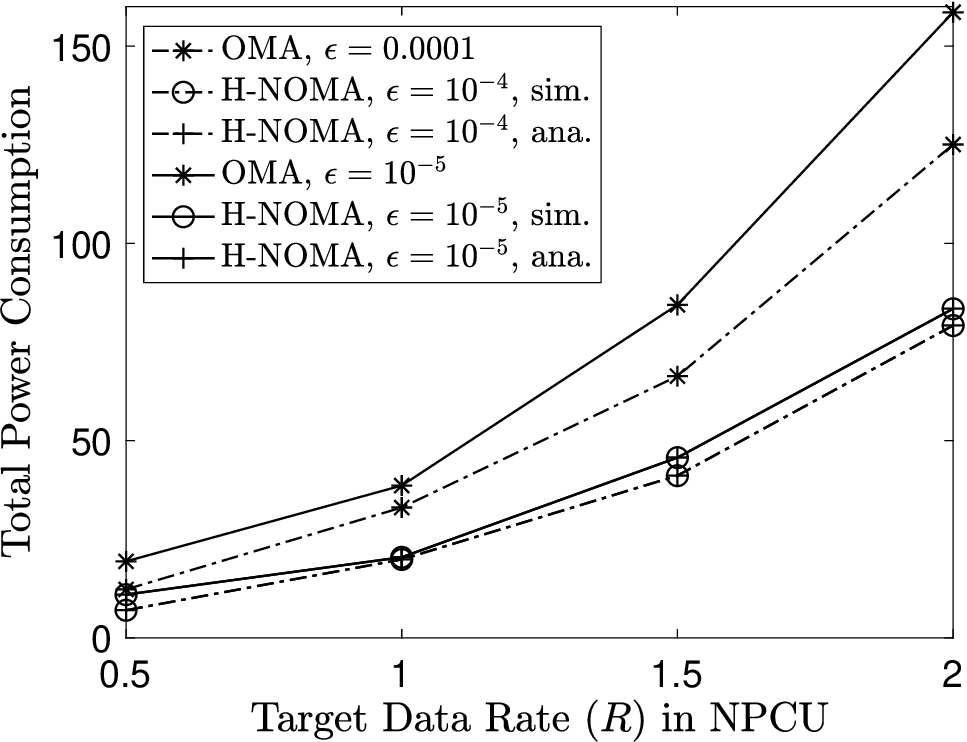, width=0.45\textwidth, clip=}\vspace{-0.5em}
\caption{The performance of hybrid NOMA assisted OFDMA in the two-user special case, where the analytical results are based on Lemma \ref{lemma4}.    \vspace{-1em}    }\label{fig3}   \vspace{-0.5em} 
\end{figure}

Recall that for the two-user special case,  the optimal solution of hybrid NOMA transmission can be obtained, as shown in Lemma \ref{lemma4}. Therefore, Figs. \ref{fig3} and \ref{fig4} are provided to focus on this special case. As can be seen from the figures, the simulation results   match  the analytical results shown in Lemma \ref{lemma4}, which verifies the accuracy of the developed analytical results and confirms the optimality of the closed-form solutions shown in Lemma \ref{lemma4}. 
Recall that in Figs. \ref{fig1} and \ref{fig2}, the users' channel gains are clipped to avoid the energy singularity issue, i.e., $h_{m,n}\geq \epsilon$ is assumed.  Fig. \ref{fig3} is also provided to demonstrate the impact of this channel clipping on the performance of the considered transmission schemes. As can be seen from the figure, the use of smaller $\epsilon$ increases the performance gain of hybrid NOMA over OMA. Or in other words, this channel clipping assumption is more critical to OMA than hybrid NOMA. The reason for this observation is that  in hybrid NOMA, a user has access to more subcarriers and hence is more robust to   deep fading, where the use of deeply faded subcarriers 
   can be avoided.  

      \begin{figure}[t]\centering \vspace{-0em}
    \epsfig{file=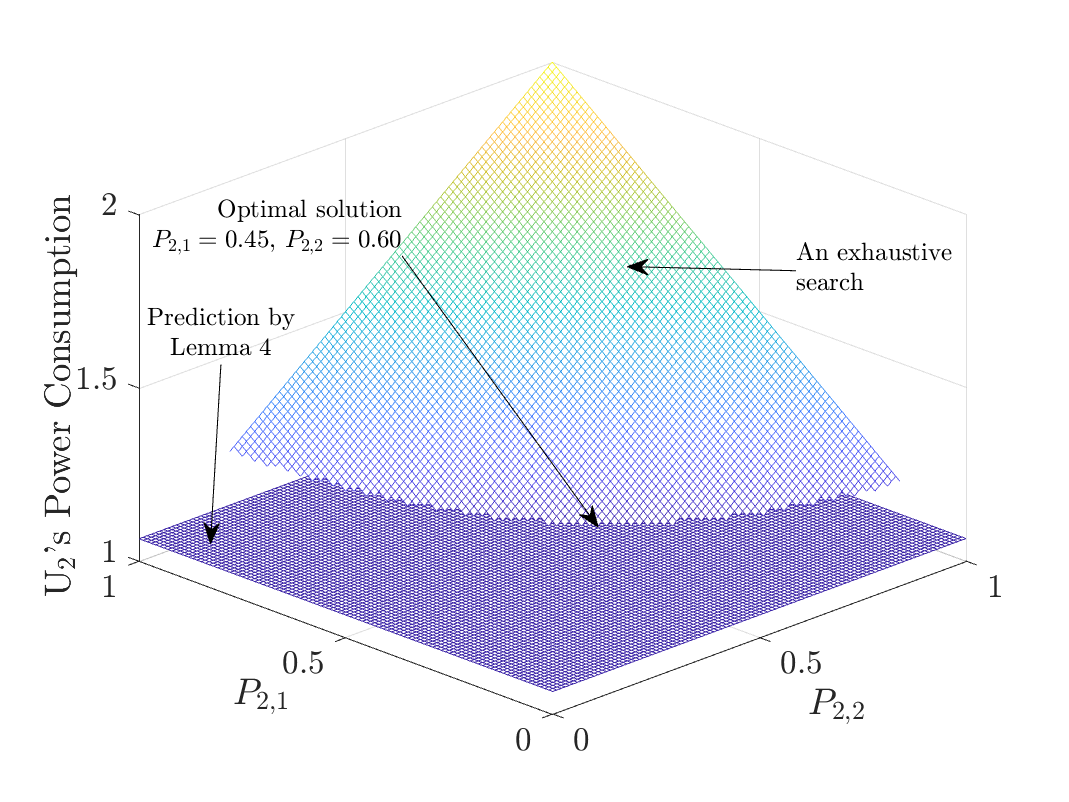, width=0.45\textwidth, clip=}\vspace{-0.5em}
\caption{Comparison between the solutions obtained from Lemma \ref{lemma4} and an exhaustive search, where  $\epsilon=10^{-5}$,  $M=2$ and $R=1$ NPCU.    \vspace{-1em}    }\label{fig4}   \vspace{-1em} 
\end{figure}
\vspace{-1em}
\subsection{Power Outage Probability}
The two-user special case is considered  first, because a closed-form expression for the power outage probability can be obtained in this special case. In Fig. \ref{fig5}, only ${\rm U}_2$'s power outage performance is shown, since ${\rm U}_1$ experiences the same performance in OMA and NOMA. As can be seen from the figure, the use of hybrid NOMA can significantly improve the power outage performance, compared to OMA, particularly at high SNR. In addition, Fig. \ref{fig5} demonstrates that the analytical results   match the simulation results at high SNR, which verifies the accuracy of Lemma \ref{lemma5}. Furthermore,  the slope of the user's power outage probability in NOMA is larger than OMA, which confirms Corollary \ref{corollary2}, i.e., the power diversity order of hybrid NOMA is larger than that of OMA.

The behavior  of hybrid NOMA shown in   Fig. \ref{fig5} is closely related  to that seen in  Fig. \ref{fig4}, and together they reveal the key benefit to use hybrid NOMA. Recall that Fig. \ref{fig4} shows that the use of a smaller $\epsilon$ increases the performance gain of hybrid NOMA over   OMA. Fig. \ref{fig5} shows that  the use of hybrid NOMA achieves a larger power diversity gain than OMA. These performance gains are due to the key feature of hybrid NOMA that a user is offered more transmission opportunities,  e.g., subcarriers, compared to   OMA. Again take the two-user scenario as an example. With OMA, ${\rm U}_2$ has to rely on a single subcarrier for its transmission, which means that its transmit power has to be very large if it suffers deep fading at the assigned subcarrier. With hybrid NOMA, the user has access to two subcarriers. If the user suffers deep fading at one subcarrier, it can discard this subcarrier and use the other one, which is the reason for the substantial   performance gain of hybrid NOMA at small $\epsilon$ shown in Fig. \ref{fig4} and the large power diversity gain shown in Fig. \ref{fig5}. 

      \begin{figure}[t]\centering \vspace{-0em}
    \epsfig{file=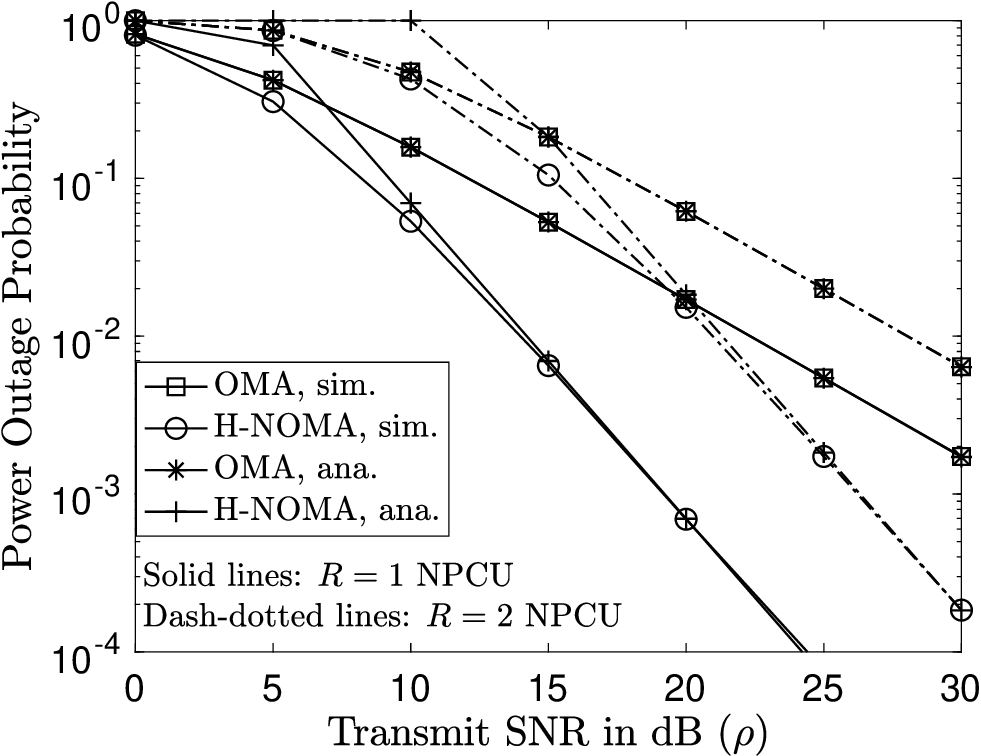, width=0.45\textwidth, clip=}\vspace{-0.5em}
\caption{${\rm U}_2$'s power outage probability achieved by hybrid NOMA assisted OFDMA, where the two-user special case is considered.  The analytical results are based on Lemma \ref{lemma5}.   \vspace{-1em}    }\label{fig5}   \vspace{-1em} 
\end{figure}


     \begin{figure}[t] \vspace{-0em}
\begin{center}
\subfigure[$R=1$ NPCU]{\label{fig6a}\includegraphics[width=0.45\textwidth]{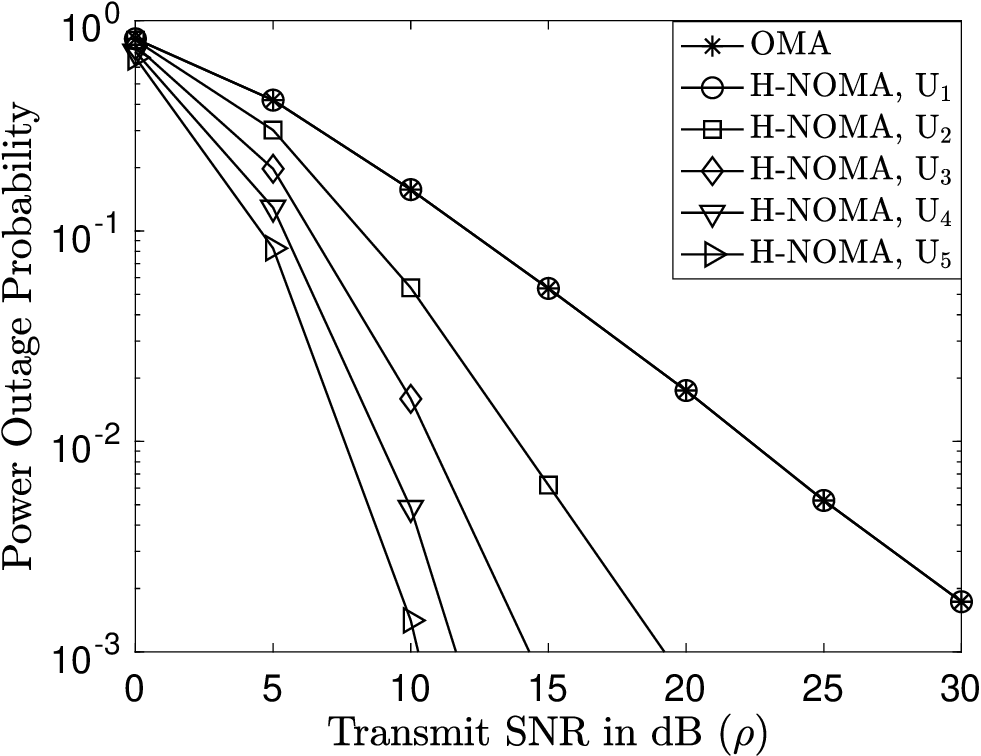}} 
\subfigure[ $R=2$ NPCU]{\label{fig6b}\includegraphics[width=0.45\textwidth]{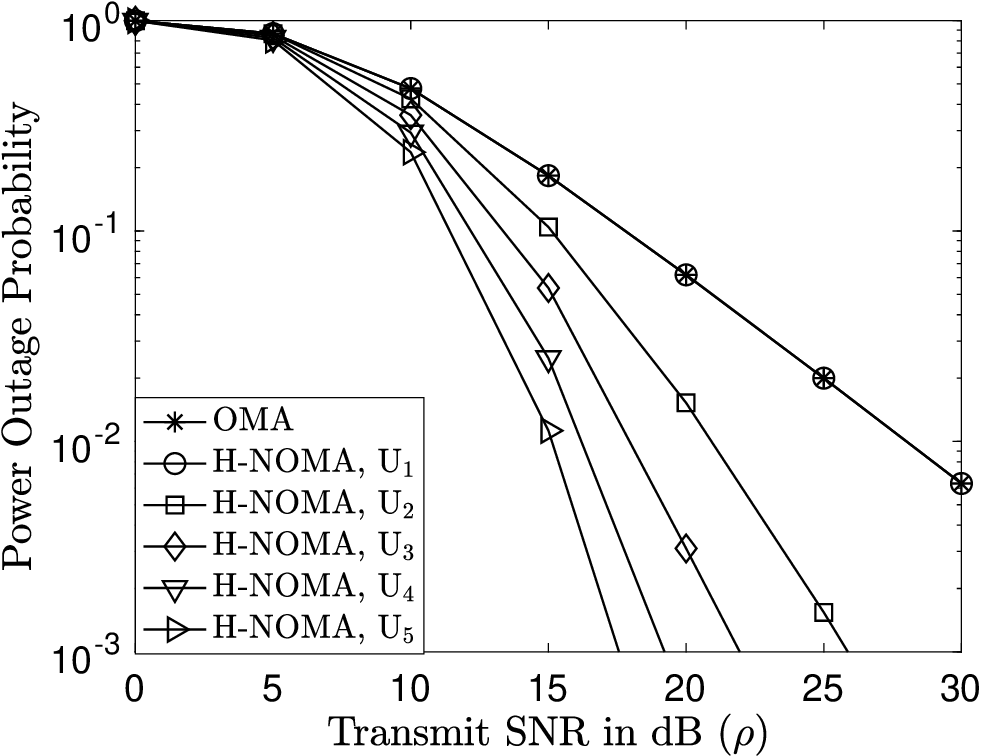}} \vspace{-1em}
\end{center}
\caption{The users' individual power outage performance  in a general multi-user scenario, where $M=5$.     \vspace{-1em} }\label{fig6}\vspace{-1em}
\end{figure}

In Fig. \ref{fig6}, a general multi-user scenario is considered, where the users' individual power outage probabilities are shown as   functions of the transmit SNR. Because the users' channel gains are i.i.d., the users' OMA power outage probabilities are the same, which is the reason why there is a single curve for the OMA case. As can be seen from the figure, with hybrid NOMA,  the users' power outage probabilities are appropriate for their energy profiles. In particular, recall that    ${\rm U}_m$ is assumed to be more energy constrained than ${\rm U}_n$, $n<m$. Fig. \ref{fig6a} shows that for the case with $R=1$ NPCU, to realize a power outage probability of $10^{-2}$, ${\rm U}_1$, the user that  is the least energy constrained, requires a transmit SNR of $23$ dB, whereas ${\rm U}_5$, the user that is  the most energy constrained, requires a transmit SNR of $7$ dB only, i.e.,   for this considered case, ${\rm U}_5$'s transmit power is just   one-fifth of ${\rm U}_1$'s transmit power. Similar performance gains can also be observed from Fig. \ref{fig6b} for the case with $R=2$ NPUC.  It is   worthy pointing   out that   by applying the proposed hybrid NOMA scheme, ${\rm U}_1$'s achievable data rate is the same as that of OMA, which is the reason why in Fig. \ref{fig6}, ${\rm U}_1$'s power outage probabilities in OMA and NOMA are the same. 
 \section{Conclusions}
In this paper,  hybrid NOMA assisted OFDMA uplink transmission has been  investigated. In particular, the impact of the unique feature of H-NOMA-OFDMA, i.e., the availability of dynamic user channel conditions, on the system performance has been  analyzed from   two perspectives. From the optimization perspective,   analytical results have been  developed to show that with H-NOMA-OFDMA,  the pure OMA mode is rarely  adopted by the users, and the pure NOMA mode could be optimal for minimizing the users' energy consumption, which are different from those  conclusions   made in earlier studies of hybrid NOMA based on TDMA. From  the statistical  perspective, two new metrics, namely the power outage probability and the power diversity gain, have been developed to    measure the performance gain of hybrid NOMA over OMA. The developed analytical results also demonstrate the capability of hybrid NOMA to meet   users' diverse energy profiles. 
 \appendices
 \section{Proof for Lemma \ref{lemma1}} \label{proof1}
 
 Recall that the KKT conditions are   necessary   for any   optimal solutions, regardless of whether the optimization problem is convex. Therefore, the lemma can be proved by first finding the KKT conditions of problem \eqref{pb:1} and then obtaining the conditions for the pure OMA solution to satisfy the KKT conditions.

Note that the Lagrangian  of problem \eqref{pb:1} can be obtained as follows: 
\begin{align}
L=&  \sum^{M}_{m=1}  \sum^{m}_{n=1}P_{m,n} +\sum^{M}_{m=1} \lambda_m
  \left(R-  \sum^{m}_{n=1} R_{m,n} \right)\\\nonumber &- \sum^{M}_{m=1}  \sum^{m}_{n=1}\lambda_{m,n}P_{m,n},
\end{align}
where $\lambda_m$ is the Lagrange multiplier for the constraint $   \sum^{m}_{n=1} R_{m,n}\geq R$, $1\leq m \leq M$, and $\lambda_{m,n}$  is the Lagrange multiplier for the constraint $P_{m,n}\geq 0$.  For problem \eqref{pb:1}, the corresponding KKT conditions can be expressed as follows: 
\begin{align}
\left\{\begin{array}{l}  
\frac{\partial L}{\partial P_{m,n}}  =0, \quad n\leq m \leq M, 1\leq m \leq M\\
 \lambda_m
  \left(R-  \sum^{m}_{n=1} R_{m,n} \right) =0,  \quad 1\leq m \leq M\\
 \lambda_{m,n}P_{m,n}=0,  \quad n\leq m \leq M, 1\leq m \leq M
 \end{array}\right.,
 \end{align}
 where the primal feasibility conditions are omitted due to space limitations. 

The challenging step to evaluate  the KKT conditions of problem \eqref{pb:1} is to find the stationarity condition, i.e., $\frac{\partial L}{\partial P_{m,n}} =0$, due to the complicated expressions for $R_{m,n}$. In particular, the stationarity conditions can be expressed as follows:
\begin{align}
&\frac{\partial L}{\partial P_{m,n}} = 1- \sum^{M}_{i=1} \lambda_i
 \frac{\partial}{\partial P_{m,n}} \sum^{i}_{j=1} R_{i,j} -\lambda_{m,n} . 
\end{align}

A key observation that can be used to simplify the stationarity conditions is that $P_{m,n}$ has an impact on $R_{i,n}$, $m\leq i \leq M$, only. Therefore,  $\frac{\partial R_{i,j} }{\partial P_{m,n}} =0$, if $j\neq n$, and $\frac{\partial R_{i,n} }{\partial P_{m,n}} =0$, if $1\leq i<m$. By using this observation,  the stationarity condition can be simplified as follows: 
\begin{align}
&\frac{\partial L}{\partial P_{m,n}} = 1- \sum^{M}_{i=m} \lambda_m\frac{\partial R_{i,n} }{\partial P_{m,n}} -\lambda_{m,n} 
\\\nonumber &=
 1-\lambda_{m,n} - \sum^{M}_{i=m} \lambda_m\frac{\partial   }{\partial P_{m,n}} \log\left(
1+
\frac{h_{i,n}P_{i,n}}{\sum^{i-1}_{j=n}h_{j,n}P_{j,n}+1}
\right),
\end{align}
for $n\leq m-1$. The special  case of $n=m$ will be discussed later. 

It is important to note that $P_{m,n}$ appears not only in the numerator of the signal-to-interference-plus-noise ratio (SINR)   of $R_{m,n}$, but also in the denominator of the SINR of $R_{m,i}$, $m+1\leq i\leq M$. To facilitate the derivative calculation,   the stationarity conditions needs to be further rewritten as follows: 
\begin{align}\nonumber 
&\frac{\partial L}{\partial P_{m,n}} =   1-  \lambda_m\frac{\partial   }{\partial P_{m,n}} \log\left(
1+
\frac{h_{m,n}P_{m,n}}{\sum^{m-1}_{j=n}h_{j,n}P_{j,n}+1}
\right)\\\label{pmln} 
&- \sum^{M}_{i=m+1} \lambda_m\frac{\partial   }{\partial P_{m,n}} \log\left(
1+
\frac{h_{i,n}P_{i,n}}{\sum^{i-1}_{j=n}h_{j,n}P_{j,n}+1}
\right)-\lambda_{m,n} . 
\end{align}
With some straightforward algebraic manipulations,   the stationarity conditions  can be expressed as follows:  
\begin{align}\label{partiall}
\frac{\partial L}{\partial P_{m,n}}  
=&   
    1-\lambda_{m,n} -  \lambda_m \frac{ 
 h_{m,n }}{1+
 \sum^{m}_{j=n}h_{j,n}P_{j,n}}\\\nonumber 
&- \sum^{M}_{i=m+1} \lambda_m
\frac{h_{m,n}}{
1+ \sum^{i}_{j=n}h_{j,n}P_{j,n} }
\\\nonumber 
&+ \sum^{M}_{i=m+1} \lambda_m\frac{h_{m,n}}{ 1+\sum^{i-1}_{j=n}h_{j,n}P_{j,n}},
\end{align}
for $n\leq m-1$. For the special case $n=m$, the stationarity conditions in \eqref{pmln} needs to be rewritten as follows:
\begin{align} \nonumber
&\frac{\partial L}{\partial P_{m,m}} =   1-\lambda_{m,m} -  \lambda_m\frac{\partial   }{\partial P_{m,m}} \log\left(
1+
 h_{m,m}P_{m,m} 
\right)\\ 
&- \sum^{M}_{i=m+1} \lambda_m\frac{\partial   }{\partial P_{m,m}} \log\left(
1+
\frac{h_{i,m}P_{i,m}}{\sum^{i-1}_{j=m}h_{j,m}P_{j,m}+1}
\right). 
\end{align}
It can be verified that the result in \eqref{partiall} is still applicable to this special case since $\frac{\partial   }{\partial P_{m,m}} \log\left(
1+
 h_{m,m}P_{m,m} 
\right) = \frac{ h_{m,m}}{1+
 h_{m,m}P_{m,m}}$

Recall that the OMA solution is given by $P_{m,m} = \frac{e^R-1}{h_{m,m}}$, and $P_{m,n}=0$ for $n< m$ and $1\leq m \leq M$.
By substituting the OMA solution into \eqref{partiall}, the stationarity conditions  can be simplified  as follows:
\begin{align}
\frac{\partial L}{\partial P_{m,n}}    
=&   1-\lambda_{m,n} -  \lambda_m \frac{ 
 h_{m,n }}{1+
 h_{n,n}P_{n,n}}\\\nonumber 
&- \sum^{M}_{i=m+1} \lambda_m
\frac{h_{m,n}}{
1+  h_{n,n}P_{n,n} }
\\\nonumber 
&+ \sum^{M}_{i=m+1} \lambda_m\frac{h_{m,n}}{ 1+ h_{n,n}P_{n,n}}
\\\nonumber
=&   1-\lambda_{m,n} -  \lambda_m \frac{ 
 h_{m,n }}{1+
 h_{n,n}P_{n,n}} , 
\end{align}
where the first step follows from the fact that    $ \sum^{i}_{j=n}h_{j,n}P_{j,n} =h_{n,n}P_{n,n}$ and  $\sum^{i-1}_{j=n}h_{j,n}P_{j,n}= h_{n,n}P_{n,n}$ because  $n\leq m\leq i-1$. 

Therefore, with the pure OMA solution,  the KKT conditions can be reduced to the following:
\begin{align}
\left\{\begin{array}{l}  
1-\lambda_{m,n} -  \lambda_m \frac{ 
 h_{m,n }}{1+
 h_{n,n}P_{n,n}} =0,   n\leq m, 1\leq m \leq M\\
 \lambda_m
  \left(R-  \sum^{m}_{n=1} R_{m,n} \right) =0,    1\leq m \leq M\\
 \lambda_{m,n}P_{m,n}=0,    n\leq m \leq M, 1\leq m \leq M
 \end{array}\right.. 
 \end{align}
 Note that  the pure OMA solution satisfies the primary feasibility conditions, and ensures $R-  \sum^{m}_{n=1} R_{m,n}=0$, $1\leq m\leq M$. Therefore, the use of the pure OMA solution simplifies the KKT conditions to the following form: 
 \begin{align}
\left\{\begin{array}{l}  
1-\lambda_{m,n} -  \lambda_m \frac{ 
 h_{m,n }}{1+
 h_{n,n}P_{n,n}} =0,   n\leq m, 1\leq m \leq M\\ 
 \lambda_{m,n}P_{m,n}=0,    n\leq m \leq M, 1\leq m \leq M
 \end{array}\right.. 
 \end{align}
By using the fact    that $P_{m,m} = \frac{e^R-1}{h_{m,m}}$ means that $h_{m,m}P_{m,m}+1 = e^R $, the KKT conditions can be further simplified as follows: 
  \begin{align}\label{xe1}
\left\{\begin{array}{l}  
1-\lambda_{m,n} - e^{-R} \lambda_m  
 h_{m,n } =0,   n\leq m , 1\leq m \leq M\\ 
 \lambda_{m,n}P_{m,n}=0,    n\leq m \leq M, 1\leq m \leq M
 \end{array}\right.. 
 \end{align}

 Recall that  for the OMA case,  $P_{m,m} = \frac{e^R-1}{h_{m,m}}$, and $P_{m,n}=0$ for $n< m$ and $1\leq m \leq M$, which means $\lambda_{m,m} =0$ for $1\leq m \leq M$. Therefore, the first part of the KKT conditions in  \eqref{xe1} can be expressed as follows:
  \begin{align} 
1  - e^{-R} \lambda_m  
 h_{m,m } =0, \quad  1\leq m \leq M , 
 \end{align}
 which means that the Lagrange multipliers corresponding to the constraint $ \sum^{m}_{n=1} R_{m,n}\geq R$, $\lambda_m$, can be obtained as follows:
   \begin{align} 
\lambda_m  = \frac{e^{R} }{ h_{m,m } } 
 , \quad  1\leq m \leq M ,
 \end{align}
 which are always positive. 
The Lagrange multipliers corresponding to   zero power allocation coefficients can be obtained as follows:
   \begin{align} 
\lambda_{m,n} =1- e^{-R} \lambda_m  
 h_{m,n }   = 1-  \frac{  h_{m,n }  }{ h_{m,m } } ,
 \end{align}
 where $1\leq n\leq m-1$ and $1\leq m \leq M$. Therefore, a necessary condition for the pure OMA solution to be optimal is given by  
    \begin{align} 
  1-  \frac{  h_{m,n }  }{ h_{m,m } } \geq 0\rightarrow \frac{  h_{m,n }  }{ h_{m,m }}\leq 1, \quad 1\leq n\leq m-1.
 \end{align}
 The proof of the lemma is complete. $\qed$
 
 \section{Proof for Lemma \ref{lemma2}}\label{proof2}
 
By following the steps in the proof of Lemma \ref{lemma1},  for the considered pure NOMA solution,   the stationarity conditions can be simplified as follows:
 \begin{align}
\frac{\partial L}{\partial P_{m,n}}  
=&      1-\lambda_{m,n} -  \lambda_m \frac{ 
 h_{m,n }}{1+ h_{n,n}P_{n,n}} .
\end{align} 
Recall that the KKT conditions for problem \eqref{pb:1} are given by
\begin{align}\label{xer1}
\left\{\begin{array}{l}  
\frac{\partial L}{\partial P_{m,n}}  =0, \quad n\leq m \leq M, 1\leq m \leq M\\
 \lambda_m
  \left(R-  \sum^{m}_{n=1} R_{m,n} \right) =0,  \quad 1\leq m \leq M\\
 \lambda_{m,n}P_{m,n}=0,  \quad n\leq m \leq M, 1\leq m \leq M
 \end{array}\right..
 \end{align}
 
By substituting the pure NOMA solution into \eqref{xer1}, the  KKT conditions can be simplified as follows:
\begin{align}
\left\{\begin{array}{l}  
1-\lambda_{m,n} - e^{-R} \lambda_m 
 h_{m,n }  =0, \quad n\leq m \leq M, n\neq  m \neq M \\
 1-\lambda_{M,M}- \lambda_M h_{M,M}=0\\
 \lambda_{m,n}P_{m,n}=0,  \quad n\leq m \leq M, 1\leq m \leq M
 \end{array}\hspace{-1em}\right..
 \end{align}
 Recall that the considered pure NOMA solution is $P_{m,m}=\frac{e^R-1}{h_{m,m}}$, $1\leq m \leq M-1$, $P_{M,1}=\frac{e^R(e^R-1)}{h_{M,1}}$, and the other power coefficients are zero.  Therefore,  for $\lambda_{m,m}=0$,   $1\leq m \leq M-1$, the Lagrange multipliers corresponding to  $ \sum^{m}_{n=1} R_{m,n}\geq R$, $1\leq m\leq M-1$, are given by 
 \begin{align} 
1 - e^{-R} \lambda_m 
 h_{m,m }  =0\rightarrow     \lambda_m 
 =\frac{e^R}{ h_{m,m } },
  \end{align}
 which are always positive.  The Lagrange multipliers corresponding to  $ \sum^{M}_{n=1} R_{M,n}\geq R$, $\lambda_M$, and the one corresponding to $P_{M,M}$, $\lambda_{M,M}$, are coupled as follows:
  \begin{align}
\left\{\begin{array}{l}  
 1-\lambda_{M,M}- \lambda_M h_{M,M}=0  \\
1  - e^{-R} \lambda_M 
 h_{M,1 }  =0
 \end{array}\right.,
 \end{align}
 where the second equation follows from the fact that $\lambda_{M,1}=0$ since $P_{M,1}\neq 0$. By solving the above two equations, the two multipliers can be obtained as follows:
   \begin{align}\label{lem1}
 \lambda_{M,M}= 1- \frac{e^R}{ h_{M,1 } } h_{M,M}, \quad 
  \lambda_M 
 =\frac{e^R}{ h_{M,1 } }
 \end{align}
 Following     steps similar to those in the proof of Lemma \ref{lemma1}, the other multipliers can be obtained as follows: 
 \begin{align}\label{lem2}
\left\{\begin{array}{l}  
 \lambda_{m,n} =   1-\frac{  h_{m,n }  }{ h_{m,m } }
 , \quad 1\leq n\leq m \leq M-1 \\
\lambda_{M,n} =1-   \frac{ h_{M,n } }{ h_{M,1 } }
 ,  \quad 2\leq n\leq M-1
 \end{array}\right.. 
 \end{align}
 Combining \eqref{lem1} and \eqref{lem2}, the conditions shown in the lemma can be obtained, which concludes the proof. $\qed$
 
 \section{Proof for Lemma \ref{lemma4}}\label{proof4}
 Recall that for the two-user special case,  problem \ref{pb:6} can be simplified as follows:
  \begin{problem}\label{pb:8} 
  \begin{alignat}{2}
\underset{P_{2,n}\geq0   }{\rm{min}} &\quad  P_{2,1}+P_{2,2}\label{8tst:1}
\\ s.t. &\quad \log\left(
1+
\frac{h_{2,1}P_{2,1}}{ h_{1,1}P_{1,1}+1}
\right)+\log\left(
1+
 h_{2,2}P_{2,2} 
\right)\geq R .  \nonumber 
  \end{alignat}
\end{problem}  
By using the fact that $P_{1,1}^{\rm SRA} = \frac{e^R-1}{h_{1,1}}$, problem \eqref{pb:8} can be simplified as follows:
 \begin{problem}\label{pb:9} 
  \begin{alignat}{2}
\underset{P_{2,n}\geq0   }{\rm{min}} &\quad  P_{2,1}+P_{2,2}\label{9tst:1}
\\ s.t. &\quad \log\left(
1+
e^{-R} h_{2,1}P_{2,1} 
\right)+\log\left(
1+
 h_{2,2}P_{2,2} 
\right)\geq R .   \label{9tst:2}
  \end{alignat}
\end{problem}

First note that the 
  Lagrangian  of problem \eqref{pb:9}  is given by
\begin{align}
L =& P_{2,1}+P_{2,2} - \sum^{3}_{i=2}\lambda_i P_{2,i-1}\\\nonumber &+
\lambda_1 \left(
R - \log\left(
1+
e^{-R} h_{2,1}P_{2,1} 
\right)-\log\left(
1+
 h_{2,2}P_{2,2} 
\right)
\right),
\end{align}
where $\lambda_i$, $1\leq i \leq 3$,  denote the Lagrange multipliers, 
and the corresponding KKT conditions are given by
\begin{align}
\left\{\begin{array}{l}  
1-\frac{\lambda_1 e^{-R} h_{2,1}}{1+
e^{-R} h_{2,1}P_{2,1} }-\lambda_2=0,\\
1 - \frac{\lambda_1  h_{2,2}}{1+
 h_{2,2}P_{2,2} }-\lambda_3=0\\
R - \log\left(
1+
e^{-R} h_{2,1}P_{2,1} 
\right)-\log\left(
1+
 h_{2,2}P_{2,2} \right)=0
 \end{array}\right.. 
 \end{align}
 It is challenging to directly establish necessary and sufficient  condition for the hybrid NOMA solution to be optimal. Therefore,   conditions for adopting the pure NOMA and OMA solutions are obtained first, as shown in the following. 
 
 \subsubsection{Pure NOMA Solution} If the pure NOMA solution is used, ${\rm U}_2$ does not use ${\rm F}_2$, which means that problem \eqref{pb:9} can be simplified as follows: 
 \begin{problem}\label{pb:10} 
  \begin{alignat}{2}
\underset{P_{m,n}\geq0   }{\rm{min}} &\quad  P_{2,1} \label{10tst:1}
\\ s.t. &\quad \log\left(
1+
e^{-R} h_{2,1}P_{2,1} 
\right) \geq R .  \label{10tst:2}  
  \end{alignat}
\end{problem} 
With some straightforward algebraic manipulations,   the pure NOMA solution is obtained as follows: 
\begin{align}
P_{2,1} =\frac{e^R(e^{R}-1)}{h_{2,1}}, \quad P_{2,2}=0.
\end{align} 
An important part of Lemma \ref{lemma4} is to establish   optimality conditions for    different modes, and an  optimality  condition for   the pure NOMA mode is established as follows. 
Because $P_{2,1} \neq 0$ and $P_{2,2}=0$,   $\lambda_2=0$ and possibly $\lambda_3\neq 0$. 
By using the closed-form expressions of the pure NOMA solution,    the KKT conditions can be rewritten as follows: 
\begin{align}
\left\{\begin{array}{l}  
1-\frac{\lambda_1 e^{-R} h_{2,1}}{1+
e^{-R} h_{2,1} \frac{e^R(e^{R}-1)}{h_{2,1}} } =0\\
1 -  \lambda_1  h_{2,2} -\lambda_3=0 
 \end{array}\right.,
 \end{align}
 which yield the following closed-form expressions for the Lagrange multipliers: 
 \begin{align}
\left\{\begin{array}{l}  
1- \lambda_1 e^{-2R} h_{2,1}  =0\rightarrow \lambda_1  =\frac{e^{2R}}{ h_{2,1}}\\
1 -  \lambda_1  h_{2,2} -\lambda_3=0\rightarrow \lambda_3 = 1 -  \frac{e^{2R}h_{2,2}}{ h_{2,1}}     
 \end{array}\right..
 \end{align}
 Because the considered optimization problem for the two-user special case is convex,    the    necessary and sufficient  condition for pure NOMA to be optimal is given by
 \begin{align}
 1 -  \frac{e^{2R}h_{2,2}}{ h_{2,1}}   \geq 0\rightarrow    \frac{h_{2,2}}{ h_{2,1}}   \leq e^{-2R}. 
 \end{align}

\subsubsection{Pure OMA Solution} Note that the pure OMA solution is simply  given by
\begin{align}
P_{2,1}=0, \quad P_{2,2} = \frac{e^{R}-1}{h_{2,2}}.
\end{align}
In order to establish the    necessary and sufficient  condition for the pure OMA solution to be optimal,   closed-form expressions for the Lagrange multipliers are needed. Because $P_{2,2} \neq 0$, 
 $\lambda_3=0$. By using  the fact $\lambda_3=0$ and the closed-form expression for the pure OMA solution, 
the KKT conditions can be simplified as follows:  
 \begin{align}
\left\{\begin{array}{l}  
1- \lambda_1 e^{-R} h_{2,1} -\lambda_2=0\\
1 - \frac{\lambda_1  h_{2,2}}{1+
 h_{2,2} \frac{e^{R}-1}{h_{2,2}}} =0 
 \end{array}\right.,
 \end{align}
which leads to the following closed-form expressions for the multipliers: 
 \begin{align}
\left\{\begin{array}{l}  
1- \lambda_1 e^{-R} h_{2,1} -\lambda_2=0\rightarrow \lambda_2 = 1- \frac{ h_{2,1}}{ h_{2,2}}    \\
1 -  e^{-R}\lambda_1  h_{2,2} =0 \rightarrow  \lambda_1  = \frac{ e^{R}}{ h_{2,2}}
 \end{array}\right.. 
 \end{align}
 Therefore,  the    necessary and sufficient  condition for the optimality of pure OMA is given by
 \begin{align}
 1- \frac{ h_{2,1}}{ h_{2,2}} \geq 0\rightarrow  \frac{ h_{2,1}}{ h_{2,2}} \leq 1\rightarrow 
 \frac{ h_{2,2}}{ h_{2,1}} \geq 1.
 \end{align}
 
 \subsubsection{Hybrid NOMA Solution}  
 Because tthe    necessary and sufficient  conditions for the optimality of the pure NOMA and OMA solutions have been established, the    necessary and sufficient  condition for hybrid NOMA to be optimal can be straightforwardly obtained as shown in the lemma. The remainder of the proof is to focus on how to obtain  the closed-form expression for the hybrid NOMA solution. By ignoring the constraints, $P_{2,n}\geq 0$, the Lagrange of problem \eqref{pb:9} can be simplified as follows: 
\begin{align}
L =& P_{2,1}+P_{2,2} \\\nonumber &+
\lambda \left(
R - \log\left(
1+
e^{-R} h_{2,1}P_{2,1} 
\right)-\log\left(
1+
 h_{2,2}P_{2,2} 
\right)\
\right),
\end{align}
where $
\lambda $ is the only multiplier corresponding to the constraint \eqref{9tst:2}. Therefore, the corresponding KKT conditions can be expressed as follows: 
\begin{align}\label{kktxxx}
\left\{\begin{array}{l}  
1-\frac{\lambda e^{-R} h_{2,1}}{1+
e^{-R} h_{2,1}P_{2,1} }=0\\
1 - \frac{\lambda  h_{2,2}}{1+
 h_{2,2}P_{2,2} }=0\\
R - \log\left(
1+
e^{-R} h_{2,1}P_{2,1} 
\right)-\log\left(
1+
 h_{2,2}P_{2,2} \right)=0
 \end{array}\right..
 \end{align}
By combining the above KKT conditions, the following equality can be established:
\begin{align}
R - \log\left(\lambda e^{-R} h_{2,1}
\right)-\log\left(\lambda  h_{2,2} \right)=0,
\end{align}
which means that the multiplier can be obtained as follows:
\begin{align} 
 \lambda =\frac{e^{R}}{ \sqrt{ h_{2,1}   h_{2,2}}}. 
\end{align}
By substituting the expression for $\lambda$ into the first two equations in \eqref{kktxxx}, the following is obtained:
\begin{align}
 1+
e^{-R} h_{2,1}P_{2,1} =&\lambda e^{-R} h_{2,1},\\
1+
 h_{2,2}P_{2,2} =&\lambda  h_{2,2} ,
 \end{align}
 which leads to a closed-form expression for the hybrid NOMA solution:  
 \begin{align}
 P_{2,1} =\frac{e^{R}}{ \sqrt{ h_{2,1}   h_{2,2}}}    -\frac{e^{R}}{ h_{2,1}},\\
P_{2,2} =\frac{e^{R}}{ \sqrt{ h_{2,1}   h_{2,2}}}    -\frac{1}{ h_{2,2}}.
 \end{align}
 The proof of the lemma is complete.  $\qed$
 
 \section{Proof for Lemma \ref{lemma5}} \label{proof5}
 Recall that all   three transmission modes, namely pure OMA, pure NOMA and hybrid NOMA, could  be adopted, which means that the power outage probability can be expressed as follows:
\begin{align}
\mathbb{P}^{\rm out}_2 =& \mathbb{P}\left(E^{\rm OMA}, \frac{ h_{2,2}}{ h_{2,1}} \geq 1  \right) \\\nonumber
&+\mathbb{P}\left(  E^{\rm NOMA}, \frac{ h_{2,2}}{ h_{2,1}}  \leq e^{-2R}  \right) 
\\\nonumber &+\mathbb{P}\left( E^{\rm Hybrid}, e^{-2R}\leq  \frac{ h_{2,2}}{ h_{2,1}}  < 1 \right) ,
\end{align}
where $E^{\rm OMA}$, $E^{\rm NOMA}$, and $E^{\rm Hybrid}$ denote the power outage events associated to the three transmission modes, namely OMA, NOMA, and hybrid NOMA, respectively.  
 The three outage cases will be analyzed separately in the following subsections.

\subsubsection{Hybrid NOMA} For notational simplicity, define $\mathbb{P}_1\triangleq \mathbb{P}\left(E^{\rm OMA}, \frac{ h_{2,2}}{ h_{2,1}} \geq 1  \right)$, which can be evaluated as follows:
\begin{align}
\mathbb{P}_1 =& \mathbb{P}\left(P_{2,1}+P_{2,2}\geq\rho,\frac{ h_{2,2}}{ h_{2,1}} \geq 1  \right) 
\\\nonumber =& \mathbb{P}\left(\frac{2e^{R}}{ \sqrt{ h_{2,1}   h_{2,2}}}    -\frac{e^{R}}{ h_{2,1}}   -\frac{1}{ h_{2,2}}\geq\rho, e^{-2R}\leq  \frac{ h_{2,2}}{ h_{2,1}}  < 1  \right). 
\end{align}
By defining $x_1=\sqrt{h_{2,1}}$ and $x_2=\sqrt{h_{2,2}}$, and with some algebraic manipulations, the outage probability can be expressed as follows:
\begin{align}
\mathbb{P}_1 =&    \mathbb{P}\left(  ( \rho x_2^2+1) x_1^2- 2e^{R}x_2x_1     +   e^{R} x_2^2 \leq 0 , \right.\\\nonumber &\left.  x_{1} e^{-R}\leq   x_{2}  < x_{1} \right).
\end{align}

Note that the roots of $  ( \rho x_2^2+1) x_1^2- 2e^{R}x_2x_1     +   e^{R} x_2^2 = 0$ are given by 
\begin{align}
\frac{e^{R}x_2 \pm  x_2\sqrt{e^{2R}  -e^{R}  ( \rho x_2^2+1)} }{  \rho x_2^2+1}\geq 0,
\end{align}
where  the following additional constraint is required:  
\begin{align}
e^{2R}  -e^{R}  ( \rho x_2^2+1)\geq 0.
\end{align} We note that the condition for the optimality of hybrid NOMA can be expressed as follows: 
\begin{align}
 x_{1} e^{-R}\leq   x_{2}  < x_{1}  \rightarrow x_2\leq x_1\leq e^Rx_2 . 
\end{align}
Therefore, the considered outage probability $\mathbb{P}_1$ can be expressed as follows:
\begin{align}
\mathbb{P}_1
=& \mathbb{P}\left(\frac{e^{R}x_2 -  x_2\sqrt{e^{2R}  -e^{R}  ( \rho x_2^2+1)} }{  \rho x_2^2+1}\right.
\\\nonumber &\left. \leq x_1\leq \frac{e^{R}x_2 +  x_2\sqrt{e^{2R}  -e^{R}  ( \rho x_2^2+1)} }{  \rho x_2^2+1}, \right.\\\nonumber
&\left.   e^{2R}  -e^{R}  ( \rho x_2^2+1)\geq 0,   x_2\leq x_1\leq e^Rx_2  \right).
\end{align}

By defining $f_1(x_2)= \frac{e^{R}x_2 -  x_2\sqrt{e^{2R}  -e^{R}  ( \rho x_2^2+1)} }{  \rho x_2^2+1}$ and $f_2(x_2)=\frac{e^{R}x_2 +  x_2\sqrt{e^{2R}  -e^{R}  ( \rho x_2^2+1)} }{  \rho x_2^2+1}$,    the   outage probability   can be simplified as follows:
\begin{align}\nonumber
\mathbb{P}_1
= & \mathbb{P}\left(    f_1(x_2)\leq   x_1\leq f_2(x_2),   e^{2R}  -e^{R}  ( \rho x_2^2+1)\geq 0,\right.\\\nonumber
&\left.  x_2\leq x_1\leq e^Rx_2  \right)
\\\nonumber
=&     \mathbb{P}\left(    f_1^2(x_2)\leq   x_1^2\leq f_2^2(x_2)  \right.\\\nonumber &\left.     x_2^2\leq \frac{e^{R}-1}{\rho} ,  x_2^2\leq x_1^2\leq e^{2R}x_2^2  \right).
\end{align}
By using the assumption that all the users' channel gains are complex Gaussian distributed with zero mean and unit variance, the outage probability can be evaluated as follows: 
\begin{align}
\mathbb{P}_1
= & \int^{\frac{e^{R}-1}{\rho}}_{0}\hspace{-1em}\left(
e^{-\max\{z,f_1^2(\sqrt{z})\}} -e^{-\min\left\{
e^{2R}z, f_2^2(\sqrt{z})
\right\}} 
\right)e^{-z}dz .
\end{align}
By defining $y=\rho z+1$, the above integral can be simplified as follows: 
\begin{align}
\mathbb{P}_1
= & \int^{e^R}_{1}\left(
e^{-\max\left\{\frac{y-1}{\rho} ,f_1^2\left(\sqrt{\frac{y-1}{\rho} }\right)\right\}}\right.\\\nonumber &\left. -e^{-\min\left\{
e^{2R}\frac{y-1}{\rho} , f_2^2\left(\sqrt{\frac{y-1}{\rho} }\right)
\right\}} 
\right)e^{-\frac{y-1}{\rho}}\frac{dy}{\rho} . 
\end{align}
Recall that $   f_2^2\left(\sqrt{z}\right)$ is given by
\begin{align}
 f_2^2\left(\sqrt{z}\right) =  z\frac{2e^{2R}     -e^{R}  ( \rho z+1) +2e^{R}\sqrt{e^{2R}  -e^{R}  ( \rho z+1)} }{ ( \rho z+1)^2},
\end{align}
which means that $f_2^2\left(\sqrt{\frac{y-1}{\rho} }\right)$ can be expressed as follows:
\begin{align}\label{f2y}
 f_2^2\left(\sqrt{\frac{y-1}{\rho} }\right) =  \frac{y-1}{\rho} \frac{2e^{2R}     -e^{R} y +2e^{R}\sqrt{e^{2R}  -e^{R}  y} }{ y^2}.
\end{align}
An important observation from \eqref{f2y} is that if $1\leq y\leq e^R$, $ f_2^2\left(\sqrt{\frac{y-1}{\rho} }\right) \rightarrow 0$ at high SNR, i.e., $\rho\rightarrow \infty$. Similarly, $ f_1^2\left(\sqrt{\frac{y-1}{\rho} }\right) \rightarrow 0$ if $\rho\rightarrow \infty$. Therefore, the outage probability can be approximated at high SNR as follows:
\begin{align}
\mathbb{P}_1
\approx & \int^{e^R}_{1}\left(\min\left\{
e^{2R}\frac{y-1}{\rho} , f_2^2\left(\sqrt{\frac{y-1}{\rho} }\right)
\right\}
 \right.\\\nonumber &\left. -\max\left\{\frac{y-1}{\rho} ,f_1^2\left(\sqrt{\frac{y-1}{\rho} }\right)\right\}  
\right) \frac{dy}{\rho} ,
\end{align}
which can be further simplified    as follows:
\begin{align}\label{p1}
\mathbb{P}_1 
\approx & \int^{e^R}_{1}\frac{y-1}{\rho} \left(\min\left\{
e^{2R} ,g_2(y) 
\right\}
 \right.\\\nonumber &\left. -\max\left\{1 ,g_1(y)\right\}  
\right) \frac{dy}{\rho} .
\end{align}

\subsubsection{Pure NOMA} Define $\mathbb{P}_2\triangleq \mathbb{P}\left(  E^{\rm NOMA}, \frac{ h_{2,2}}{ h_{2,1}}  \leq e^{-2R}  \right) $, which can be evaluated as follows: 
\begin{align}
\mathbb{P}_2=& \mathbb{P}\left(P_{2,1}+P_{2,2}\geq\rho,  \frac{ h_{2,2}}{ h_{2,1}}  \leq e^{-2R} \right)  
\\\nonumber
\overset{(a)}{=}& \mathbb{P}\left(\frac{e^R(e^{R}-1)}{h_{2,1}} \geq\rho,  \frac{ h_{2,2}}{ h_{2,1}}  \leq e^{-2R} \right)  
\\\nonumber
=& \mathbb{P}\left(  e^{2R} h_{2,2}  \leq   h_{2,1}\leq \frac{e^R(e^{R}-1)}{\rho}  \right)  ,
\end{align}
where step (a) follows by using the pure OMA power allocation coefficients. Again by using the complex Gaussian distribution, the outage probability can be obtained as follows: 
\begin{align}
\mathbb{P}_2
=&\int_{0}^{\frac{e^{-R}(e^{R}-1)}{\rho} }\left(e^{ -e^{2R} x}- e^{- \frac{e^R(e^{R}-1)}{\rho} }\right)e^{-x}dx
\\\nonumber =& \frac{1-e^{ -\left(1+e^{2R} \right)\frac{e^{-R}(e^{R}-1)}{\rho} }}{1+e^{2R}}- e^{- \frac{e^R(e^{R}-1)}{\rho} }\left(1-e^{-\frac{e^{-R}(e^{R}-1)}{\rho} } \right).
\end{align}
At high SNR, the outage probability can be approximated as follows: 
 \begin{align}\nonumber
\mathbb{P}_2
\approx & \frac{ \left(1+e^{2R} \right)\frac{e^{-R}(e^{R}-1)}{\rho} }{1+e^{2R}}-\frac{ \left(1+e^{2R} \right)^2\frac{e^{-2R}(e^{R}-1)^2}{\rho^2} }{2(1+e^{2R})}\\\nonumber &- \left(1 - \frac{e^R(e^{R}-1)}{\rho}  \right)\\\nonumber &\times \left(\frac{e^{-R}(e^{R}-1)}{\rho} -\frac{e^{-2R}(e^{R}-1)^2}{2\rho^2} \right)
\\\nonumber\approx&    - \left(1+e^{2R} \right)\frac{e^{-2R}(e^{R}-1)^2}{2\rho^2}  \\\nonumber &+
\frac{e^{-2R}(e^{R}-1)^2}{2\rho^2} +
  \frac{e^R(e^{R}-1)}{\rho}  \frac{e^{-R}(e^{R}-1)}{\rho}  
  \\\label{p2}
= &    
  \frac{ (e^{R}-1)^2}{2\rho^2}    .
\end{align}
 
 \subsubsection{Pure OMA} Define $ \mathbb{P}_3\triangleq \mathbb{P}\left(E^{\rm OMA}, \frac{ h_{2,2}}{ h_{2,1}} \geq 1  \right)$, which can be expressed as follows: 
 \begin{align}
\mathbb{P}_3 =& \mathbb{P}\left(P_{2,1}+P_{2,2}\geq\rho, \frac{ h_{2,2}}{ h_{2,1}} \geq 1 \right)   
\\\nonumber =& \mathbb{P}\left(\frac{e^{R}-1}{\rho }\geq h_{2,2}  \geq h_{2,1} \right)   . 
\end{align}
Again by applying the complex Gaussian distribution, the outage probability can be evaluated as follows: 
 \begin{align}
\mathbb{P}_3 =&  \int^{\frac{e^{R}-1}{\rho }}_{0}  \left(e^{-x}-e^{-\frac{e^{R}-1}{\rho }}
\right)e^{-x} dx\\\nonumber 
 =&   \frac{1}{2}\left(1-e^{-2\frac{e^{R}-1}{\rho }}\right)-e^{-\frac{e^{R}-1}{\rho }}\left(1-e^{-\frac{e^{R}-1}{\rho }}
\right)  . 
\end{align}
At high SNR, $\mathbb{P}_3$ can be approximated as follows: 
 \begin{align}\label{p3}
\mathbb{P}_3 \approx&    \frac{1}{2} \left(2\frac{e^{R}-1}{\rho }-\frac{1}{2}2^2\frac{(e^{R}-1)^2}{\rho^2 }\right) \\\nonumber &-\left(1-\frac{e^{R}-1}{\rho }\right)\left(\frac{e^{R}-1}{\rho }-\frac{1}{2}\frac{(e^{R}-1)^2}{\rho^2 }\right)
\\\nonumber \approx&  \frac{1}{2}  \frac{(e^{R}-1)^2}{\rho ^2}  .
\end{align}
It is interesting to observe that the high SNR approximations for the pure OMA and pure NOMA cases are the  same. 

By combining \eqref{p1}, \eqref{p2} and \eqref{p3}, the overall power outage probability can be approximated at high SNR as follows:
 \begin{align}
\mathbb{P}^{\rm out} \approx&      \frac{(e^{R}-1)^2}{\rho ^2}   +\int^{e^R}_{1}\frac{y-1}{\rho} \left(\min\left\{
e^{2R} ,g_2(y) 
\right\}
 \right.\\\nonumber &\left. -\max\left\{1 ,g_1(y)\right\}  
\right) \frac{dy}{\rho}.
\end{align}
 The proof of the lemma is complete.  $\qed$

\bibliographystyle{IEEEtran}
\bibliography{IEEEfull,trasfer}

  \end{document}